%% file: 9710_final.tex
\documentclass{fundam}

\input{macros}
\publyear{22}
\papernumber{2134}
\volume{187}
\issue{2-4}

 \finalVersionForARXIV

\begin{document}

                             \title{A Polyhedral Abstraction for Petri Nets and\\
                                     its Application to SMT-Based Model Checking}

\author{Nicolas Amat\thanks{Address for correspondence: LAAS-CNRS, Vertics, 7 ave.
                                Colonel Roche, 31031 Toulouse cedex 4, France}
  \\
  LAAS-CNRS\\
  Universit\'e de Toulouse, CNRS, INSA, Toulouse, France\\
  nicolas.amat{@}laas.fr
  \and Bernard Berthomieu,\ Silvano Dal~Zilio\\
  LAAS-CNRS\\
  Universit\'e de Toulouse, CNRS, Toulouse, France
  }
\maketitle
\runninghead{N. Amat et al.}{Polyhedral Abstraction for Model Checking Petri Nets}
\begin{abstract}
  We define a new method for taking advantage of net reductions in
  combination with a SMT-based model checker. Our approach consists in
  transforming a reachability problem about some Petri net, into the
  verification of an updated reachability property on a reduced
  version of this net. This method relies on a new state space
  abstraction based on systems of constraints, called polyhedral
  abstraction.

  We prove the correctness of this method using a new notion of
  equivalence between nets. We provide a complete framework to define
  and check the correctness of equivalence judgements; prove that this
  relation is a congruence; and give examples of basic equivalence
  relations that derive from structural reductions.

  Our approach has been implemented in a tool, named SMPT,
  that provides two main procedures: Bounded Model Checking (BMC) and
  Property Directed Reachability (PDR). Each procedure has been
  adapted in order to use reductions and to work with arbitrary Petri
  nets. We tested SMPT on a large collection of queries used
  in the Model Checking Contest. Our experimental results show that
  our approach works well, even when we only have a moderate amount of
  reductions.
\end{abstract}



\section{Introduction}\vspace*{-1mm}

A significant focus in model checking research is finding algorithmic
solutions to avoid the ``state explosion problem'', that is finding
ways to analyse models that are out of reach using current
verification methods.  To overcome this problem, it is often useful to
rely on symbolic representation of the state space (like with decision
diagrams) or on an abstraction of the problem, for instance with the
use of logical approaches like SAT solving. We can also benefit from
optimizations related to the underlying model. When analysing Petri
nets, for instance, a valuable technique relies on the transformation
and decomposition of nets, a method pioneered by
Berthelot~\cite{berthelot_transformations_1987} and known as
\emph{structural reduction}.

We recently proposed a new abstraction technique based on
reductions~\cite{berthomieu2018petri,berthomieu_counting_2019}. The idea is to
compute reductions of the form $(N, E, N')$, where: $N$ is an initial net (that
we want to analyse); $N'$ is a residual net (hopefully much simpler than $N$);
and $E$ is a system of linear constraints. The idea is to preserve enough
information in $E$ so that we can rebuild the reachable markings of $N$ knowing
only the ones of $N'$. In a nutshell, we capture and abstract the effect of
reductions using a set of linear constraints between the places of $N$ and $N'$.

In this paper, we show that this approach works well when combined
with SMT-based verification. In particular, it provides an elegant way
to integrate reductions into known verification procedures. To support
this statement, we provide a full theoretical framework based on the
definition of a new equivalence relation between Petri nets
(Sect.~\ref{sec:reach-equiv-net}) and show how to use it for checking
safety and invariant properties (Sect.~\ref{sec:smt-based-model}). Our
method does not impose restrictions on the syntax of nets, such as
constraints on the weights of arcs or bounds on the marking of
places.

We have previously applied this technique in a symbolic model checker,
called \textsc{Tedd}, that uses Set Decision
Diagrams~\cite{thierry2009hierarchical} in order to generate an
abstract representation for the state space of a net $N$.
In practice, we can often reduce a Petri net $N$ with $n$ places (from
a {high dimensional} space) into a residual net $N'$ with far fewer
places, say $n'$ (in a lower-dimensional space). Hence, with our
approach, we can represent the state space of $N$ as the ``inverse
image'', by the linear system $E$, of a subset of vectors of dimension
$n'$. This technique can result in a very compact representation of
the state space. We observed this effect during the recent editions of
the Model Checking Contest (MCC)~\cite{mcc2019}, where our tool
finished at the first place for three consecutive years in the
\emph{State Space} category.  In this paper, we show that we can
benefit from the same ``dimensionality reduction'' effect when using
automatic deduction procedures. Actually, since we are working with
(possibly unbounded) vectors of integers, we need to consider SMT
instead of SAT solvers. We show that it is enough in our case to use
solvers for the theory of Quantifier-Free formulas on Linear Integer
Arithmetic, what is known as QF-LIA in SMT-LIB~\cite{BarFT-RR-17}.

To adapt our approach with the theory of SMT solving, we define an
abstraction based on Boolean combinations of linear constraints
between integer variables (representing the marking of places). This
results in a new relation $N \reduc_{E} N'$, which is the counterpart
of the tuple $(N, E, N')$ in a SMT setting. We named this relation a
\emph{polyhedral abstraction} in reference to ``polyhedral models''
used in program optimization and static
analysis~\cite{besson1999polyhedral,feautrier1996automatic}. Indeed,
like in these works, we propose an algebraic representation of the
relation between a model and its state space based on the sets of
solutions to systems of constraints. We should
also often use the term \emph{$E$-abstraction equivalence} to
emphasize the importance of the linear system $E$. One of our main results is
that, given a relation $N \reduc_{E} N'$, we can derive a formula
$\tilde{E}$ such that $F$ is an invariant for $N$ if and only if
$\tilde{E} \wedge F$ is an invariant for the net $N'$.
\eject

\noindent Since the
residual net may be much simpler than the initial one, we expect that
checking the invariant $\tilde{E} \wedge F$ on $N'$ is more efficient
than checking $F$ on $N$.

Our approach has been implemented and computing experiments show that
reductions are effective on a large benchmark of queries. We provide a
prototype tool, called SMPT, that includes an adaptation of two
procedures, Bounded Model Checking (BMC)~\cite{biere_symbolic_1999}
and Property Directed Reachability
(PDR)~\cite{jhala_sat-based_2011,hutchison_understanding_2012}. Each
of these methods has been adapted in order to use reductions and to
work with arbitrary Petri nets. We tested SMPT on a large collection
of queries ($13\,265$ test cases) used during the 2020 edition of the
Model Checking Contest and participated, with our tool, in the two
reachability competitions in the MCC'2021. Our experimental results
show that our approach works well, even when we only have a moderate
amount of reductions.

\paragraph{Outline and contributions.}
The paper is organized as follows. We start by defining the notations
used in our work in Sect.~\ref{sec:petri-nets-linear}, where we rely
on a presentation of Petri net semantics that emphasizes the
relationship with the QF-LIA theory. In
Sect.~\ref{sec:reach-equiv-net}, we define our notion of polyhedral
abstraction and prove several of its properties. We give a description
of some of the structural reductions used in our approach and show how
they correspond to axioms of our polyhedral abstraction
equivalence. We also prove that polyhedral abstractions are preserved by
composition and transitivity, which gives a simple way to check the
equivalence between two complex nets. We use these results in
Sect.~\ref{sec:smt-based-model} and~\ref{sec:implementation} to
describe an adaptation of two SMT-based, model checking algorithms for
Petri nets that can take advantage of reductions and prove their
correctness. Before concluding, we report on experimental results on
an extensive collection of nets and queries. Our results are quite
promising. For example, on our benchmark, we observe that we are able
to compute twice as many results using reductions than without.

Many results and definitions were already presented in a shorter
version of the paper~\cite{pn2021}. This extended version contains
several additions that improve on the two main contributions of our
work.

Concerning our definition of a polyhedral abstraction for Petri nets,
we describe more precisely the reduction rules used in our approach
and give more detailed proofs and definitions about the properties of
our equivalence. With these additions, we give a stand-alone
definition of $E$-abstraction equivalence that provides a more
algebraic (and therefore less monolithic) approach than the one used
previously in our work about computing the number of reachable
states~\cite{berthomieu_counting_2019}. We believe that our
equivalence could be reused in other settings.

Concerning our application to SMT-based model checking. We give a
detailed description of our adaptation of the PDR procedure for
model checking Petri nets in the case of coverability
properties.


\section{Petri nets and linear arithmetic constraints}
\label{sec:petri-nets-linear}

Some familiarity with Petri nets is assumed from the reader. We recall
some basic terminology. Throughout the text, comparison $(=, \geq)$
and arithmetic operations $(-, +)$ are extended pointwise to functions
and tuples.

\begin{definition}
A \textit{Petri net} $N$ is a tuple $(P, T, \pre, \post)$ where:
\begin{itemize}
  \item $P = \{p_1, \dots, p_n\}$ is a finite set of places,
  \item $T = \{t_1, \dots, t_k\}$ is a finite set of transitions (disjoint from
  $P$),
  \item $\pre : T \rightarrow (P \rightarrow \mathbb{N})$ and $\post :
T \rightarrow (P \rightarrow \mathbb{N})$ are the pre- and post-condition
functions (also called the flow functions of $N$).
\end{itemize}
\end{definition}

\subsection{States}
\label{sec:states}

A state $m$ of a net, also called a \emph{marking},
is a mapping $m : P \rightarrow \mathbb{N}$ which assigns a number of
\emph{tokens}, $m(p)$, to each place $p$ in $P$. A marked net $(N, m_0)$ is a
pair composed from a net and an initial marking $m_0$.

A marking $m$ is $k$-{bounded} when each place has at most $k$ tokens;
property $\bigwedge_{p \in P} m(p) \le k$ is true. Likewise, a marked
Petri net $(N, m_0)$ is bounded when there is $k$ such that all
reachable markings are $k$-bounded. A net is \emph{safe} when it is
$1$-bounded.  In our work, we consider \emph{generalized} Petri nets
(in which net arcs may have weights larger than $1$) and we do not
restrict ourselves to bounded nets.

\subsection{Behaviour}
\label{sec:semantics}

A transition $t \in T$ is \textit{enabled} at marking $m \in \Nat^P$
when $m(p) \ge \pre(t,p)$ for all places $p$ in $P$. (We can
also simply write $m \geq \pre(t)$, where $\geq$ stands for
the component-wise comparison of markings.) A marking $m' \in \Nat^P$
is reachable from a marking $m \in \Nat^P$ by firing transition $t$,
denoted $m \trans{t} m'$, if: (1) transition $t$ is enabled at
$m$; and (2) $m' = m - \pre(t) + \post(t)$.

By extension, we say that a \textit{firing sequence}
$\varrho = t_1\, \dots\, t_n \in T^*$ can be fired from $m$, denoted
$m \wtrans{\varrho} m'$, if there exist markings $m_0, \dots, m_n$ such
that $m = m_0$, $m' = m_n$ and $m_i \trans{t_{i+1}} m_{i+1}$ for all
$i$ in the range $0..n-1$.

\medskip
We denote $R(N, m_0)$ the set of markings reachable from $m_0$ in
$N$:
\begin{equation}
  R(N, m_0) \defeq \{ m \mid \exists \varrho. \  m_0 \wtrans{\varrho} m\}
\end{equation}
The \emph{semantics} of a marked net is the Labeled Transition System
(LTS), with nodes in $R(N, m_0)$ and edges between states $(m, m')$
whenever $m \trans{t} m'$. We focus mostly on reachable states in our
work and will therefore seldom refer to the LTS of the net.

\subsection{Labels and observations}
\label{sec:observations-labels}

In the following, we will often consider that each transition is
associated with a label (a symbol taken from an alphabet $\Sigma$). In
this case, we assume that a net is associated with a labeling function
$l : T \to \Sigma \cup \{ \tau \}$, where $\tau$ is a special symbol
for the silent action name. Every net has a default labeling function
$l_N$ such that $\Sigma = T$ and $l_N(t) = t$ for every transition
$t \in T$.

We can extend the notion of labels to sequences of transitions in a
straightforward way. Given a relabeling function, $l$, we can extend
it into a function from $T^\star$ into $\Sigma^\star$ such that
$l(\epsilon) = \epsilon$, $l(\tau) = \epsilon$ and
$l(\varrho\, t) = l(\varrho)\, l(t)$. Given a sequence of labels
$\sigma$ in $\Sigma^\star$, we write $(N, m) \wtrans{\sigma} (N, m')$
when there is a firing sequence $\varrho$ in $T^\star$ such that
$(N, m) \wtrans{\varrho} (N, m')$ and $\sigma = l(\varrho)$. We say in
this case that $\sigma$ is an \textit{observation sequence} of the
marked net $(N, m)$.

\subsection{Graphical notations}
\label{sec:graphical-notations}

\begin{figure}[tb]
  \centering
  \includegraphics{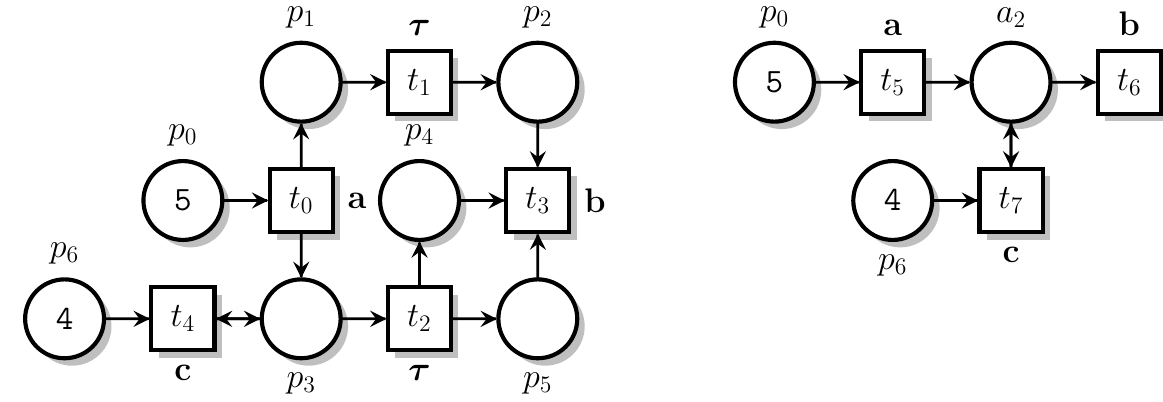}
  \caption{An example of Petri net, $M_1$ (left), and one of its
    polyhedral abstraction, $M_2$ (right), with
    $E_M \defeq (p_5 = p_4) \wedge (a_1 = p_1 + p_2) \wedge
    (a_2 = p_3 + p_4) \wedge (a_1 = a_2)$.\label{fig:stahl}}
\end{figure}

We use the standard graphical notation for nets, where places are
depicted as circles and transitions as squares. For marking, we use
the notation $\marking{{p}{k}}$ to state that place $p$ has $k$ tokens
and we omit empty places. With the net displayed in
Fig.~\ref{fig:stahl} (left), and with our convention, the initial
marking is $m_0 \defeq \marking{{p_0}{5}{p_6}{4}}$ (only $5$ and $4$
tokens in places $p_0$ and $p_6$).  We have
$m_0 \wtrans{\varrho} m'_0$ with
$\varrho \defeq t_0\,t_0\,t_1\,t_1\,t_2\,t_3\,t_4$ and
$m'_0 \defeq \marking{{p_0}{3}{p_2}{1}{p_3}{1}{p_6}{3}}$; and
therefore $m_0 \wtrans{\mathbf{a\,a\,b\,c}} m'_0$ when we only look at
(observable) labels.

\subsection{Properties}
\label{sec:properties}

We can define many properties on the markings of a net $N$ using
Boolean combinations of linear constraints with integer variables.
Assume that we have a
marked net $(N, m_0)$ with set of places $P = \{p_1, \dots, p_n\}$.
We can associate a marking $m$ over $P$ to the formula
$\underline{m}(x_1, \dots, x_n)$, below. In this context, an equation $x_i = k$ means that there must be $k$ tokens in place $p_i$.
Formula $\underline{m}$ is obviously a conjunction of literals, what is called a \emph{cube} in
\cite{jhala_sat-based_2011}.
\begin{equation}
  \underline{m}(x_1, \dots, x_n) \defeq (x_1 = m(p_1)) \wedge \dots
  \wedge (x_k = m(p_k))\label{eq:marking}
\end{equation}

In the remainder, we use the notation $\phi(\vec{x})$ for the
declaration of a formula $\phi$ with variables in $\vec{x}$, instead
of the more cumbersome notation $\phi(x_1, \dots, x_n)$. We also
simply use $\phi(\vec{v})$ instead of
$\phi\{x_1 \leftarrow v_1\}\dots\{x_n \leftarrow v_n\}$, for the
substitution of $\vec{x}$ with $\vec{v}$ in $\phi$. We often
use place names as variables (or parameters) and use $\vec{p}$ for the
vector $(p_1, \dots, p_n)$. We also often use $\underline{m}$ instead
of $\underline{m}(\vec{p})$.

\begin{definition}[Model of a formula]
  We say that a marking $m$ is a model of (or $m$ \emph{satisfies})
  property $\phi$, denoted $m \models \phi$, when formula
  $\phi(\vec{x}) \wedge \underline{m}(\vec{x})$ is satisfiable. In
  this case $\phi$ may use variables that are not necessarily in $P$.
\end{definition}

We can use this approach to reframe many properties on Petri nets. For
instance the notion of safe markings, described previously: a marking
$m$ is safe when $m \models \mathrm{BND}_1(\vec{x})$, where
$\mathrm{BND}_k$ is a predicate in QF-LIA defined as:
\begin{equation}
  \mathrm{BND}_k(\vec{x}) \defeq \bigwedge_{i \in 1..n} (x_i \leq
  k)
\end{equation}

Likewise, the property that transition $t$ is enabled corresponds to
the predicate $\mathrm{ENBL}_t$ below, in the sense that $t$ is
enabled at $m$ when $m \models \mathrm{ENBL}_t(\vec{x})$.
\begin{equation}
  \mathrm{ENBL}_t(\vec{x}) \defeq \bigwedge_{i \in 1..n} (x_i \geq
  \pre(t, p_i))
\end{equation}

Another example is the definition of \emph{deadlocks}, which are
characterized by formula
$\mathrm{DEAD}(\vec{x}) \defeq \bigwedge_{t \in T} \neg
\mathrm{ENBL}_t(\vec{x})$.
We give other examples in Sect.~\ref{sec:implementation}, when we
encode the transition relation of a Petri net using formulas.\\

In our work, we focus on the verification of \textit{safety}
properties on the reachable markings of a marked net $(N,
m_0)$. Examples of the properties that we want to check include: whether
some transition $t$ is enabled (commonly known as
\emph{quasi-liveness}); whether there is a deadlock; whether some
invariant between place markings is true; \dots

\begin{definition}[Invariant and reachable properties]
  Property $\phi$ is an invariant on $(N, m_0)$ if and only if we have
  $m \models \phi$ for all $m \in R(N,m_0)$. We say that $\phi$ is
  reachable when there exists $m \in R(N,m_0)$ such that
  $m \models \phi$.
\end{definition}

In our experiments, we consider the two main kinds of
\emph{reachability formulas} used in the MCC: $\mathrm{AG}\, \phi$
(true only when $\phi$ is an invariant), and $\mathrm{EF}\, \phi$
(true when $\phi$ is reachable), where $\phi$ is a Boolean combination
of atomic properties (it has no modalities). At various times, we will
use the fact that $\phi$ is invariant if and only if its negation is
not reachable: $\mathrm{EF}\, \neg \phi$ is false.


\section{Polyhedral abstraction and $E$-equivalence}
\label{sec:reach-equiv-net}

We define a new notion, called \textit{$E$-abstraction}, that is used
to state a correspondence between the set of reachable markings of two
Petri nets ``modulo'' some system of linear constraints
$E$. Basically, we have that $(N_2, m_2)$ is an abstraction of
$(N_1, m_1)$ when, for every sequence $m_1 \wtrans{\sigma_1} m'_1$ in
$N_1$, there must exist a sequence $m_2 \wtrans{\sigma_2} m'_2$ in
$N_2$ such that $E \wedge \underline{m'_1} \wedge \underline{m'_2}$ is
satisfiable. We also ask that there exists such sequence for every
marking $m_2'$ such that
$E \wedge \underline{m'_1} \wedge \underline{m'_2}$ is
satisfiable. Therefore, knowing $E$, we can compute the reachable
markings of $N_1$ from those of $N_2$.

\medskip
We also ask for the observation sequences, $\sigma_1$ and $\sigma_2$ in this
case, to be equal. With the addition of this constraint, we prove that the
reflexive and symmetric closure of an $E$-abstraction is also a congruence,
which we call an $E$-equivalence (defined formally in Sect.~\ref{equivalence}).

We can illustrate these notions using the two nets $M_1, M_2$ in
Fig.~\ref{fig:stahl} and the linear constraint
$E_M \defeq (p_5 = p_4) \wedge (a_1 = p_1 + p_2) \wedge (a_2 = p_3 +
p_4) \wedge (a_1 = a_2)$. Recall that marking
$m'_1 \defeq \marking{{p_0}{3}{p_2}{1}{p_3}{1}{p_6}{3}}$ is reachable
in $M_1$. We also have that $E_M \wedge \underline{m'_1}$ entails
$(p_0 = 3) \wedge (p_6 = 3) \wedge (a_2 = 1) $. Hence, if we prove
that $(M_1, m_1)$ is $E_M$-equivalent to $(M_2, m_2)$, we can conclude
that the marking $m'_2 \defeq \marking{{a_2}{1}{p_0}{3}{p_6}{3}}$ is
reachable in $M_2$.

Conversely, we have several markings (exactly $4$) in $M_1$ that
correspond to the constraint
$E_M \wedge \underline{m'_2} \equiv (p_5 = p_4) \wedge (p_1 + p_2 = 1)
\wedge (p_3 + p_4 = 1) \wedge \underline{m'_2}$. All these markings
are reachable in $M_1$ using the same observation sequence
$\mathbf{a\,a\,b\,c}$. More generally, each marking $m'_2$ of $N_2$
can be associated to a convex set of markings of $N_1$, defined as the
set of positive integer solutions of $E \wedge
\underline{m'_2}$. Moreover, these sets form a partition of
$R(N_1, m_1)$. This motivates our choice of calling this relation a
\emph{polyhedral abstraction}.

While our approach does not dictate a particular method for finding
pairs of equivalent nets, we rely on an automatic approach based on the
use of \emph{structural net reductions}. When the net $N_1$ can be
reduced, we will obtain a resulting net ($N_2$) and a condition ($E$)
such that $N_2$ is a polyhedral abstraction of $N_1$. In this case,
$E$ will always be expressed as a conjunction of equality and inequation constraints
between linear combinations of integer variables (the marking of
places). This is why we should often use the term \emph{reduction
  constraints} when referring to $E$. Our goal is to transform any
reachability problem on the net $N_1$ into a reachability problem on
the (reduced) net $N_2$, which is typically much easier to check.

\subsection{Solvable systems and $E$-equivalence}
\label{equivalence}

Before defining our equivalence more formally, we need to introduce some
constraints on the condition, $E$, used to correlate the markings of two
different nets. We say that a pair of markings $(m_1,m_2)$ are \emph{compatible}
(over respective sets of places $P_1$ and $P_2$) when they have the same number
of tokens on their shared places, meaning $m_1(p) = m_2(p)$ for all $p$ in $P_1
\cap P_2$. This is a necessary and sufficient condition for formula
$\underline{m_1} \wedge \underline{m_2}$ to be satisfiable. When this is the
case, we denote $m_1 \uplus m_2$ the unique marking in $(P_1 \cup P_2)$ defined by:
\begin{equation}
  (m_1 \uplus m_2)(p) =
  \begin{cases}
    m_1(p) & \text{ if $p \in P_1$}, \\
    m_2(p) & \text{ otherwise.}
  \end{cases}
\end{equation}
Equipped with this notion, we can say that two markings $m_1, m_2$
(defined over the set of places $P_1, P_2$ of two nets $N_1$ and
$N_2$) are ``a solution'' of constraints $E$ when they are compatible
with each other and $E \wedge \underline{m_1 \uplus m_2}$ is
satisfiable.

\medskip
This leads to the notion of \emph{solvable system}, such that every
reachable marking of $N_1$ can be paired with at least one reachable
marking of $N_2$ to form a solution of $E$; and reciprocally.

\begin{definition}[Solvable system of reduction constraints]
  A system of linear constraints, $E$, is solvable for $N_1, N_2$ if
  and only if for all reachable markings $m_1$ in $N_1$ there exists
  at least one marking $m_2$ of $N_2$, compatible with $m_1$, such
  that ${m_1 \uplus m_2} \models E$, and vice versa for every
  reachable marking $m_2$ in $N_2$.
\end{definition}

In the following, when we use an $E$-abstraction equivalence between
two marked nets $(N_1, N_2)$, we ask that condition $E$ be
\emph{solvable for $N_1, N_2$} (see condition A2).
While this property is not essential for most of our results, it simplifies our
presentation and it will always be true for the reduction constraints generated
with our method. On the other hand, we do not prohibit to use variables in $E$
that are not in $P_1 \cup P_2$. Actually, such a situation will often occur in
practice, when we start to chain several reductions.

We define our notion of $E$-abstraction as an equivalence relation between the
markings reached using equal ``observation sequences''.  An $E$-abstraction
equivalence (shortened as $E$-equivalence) is an abstraction in both directions.

\begin{definition}[$E$-abstraction and $E$-abstraction equivalence]
  Assume $N_1 = (P_1, T_1, \pre_1, \post_1)$ and
  $N_2 = (P_2, T_2, \pre_2, \post_2)$ are two Petri
  nets with respective labeling functions $l_1, l_2$, over the same
  alphabet $\Sigma$, and $E$ a system of linear constraints. We say that the marked net $(N_2, m_2)$ is an
  $E$-abstraction of $(N_1, m_1)$, denoted
  $(N_1, m_1) \sqsupseteq_E (N_2, m_2)$, if and only if:
  \begin{description}
  \item[(A1)] the initial markings are compatible with $E$, meaning
    $m_1 \uplus m_2 \models E$.
	
  \item[(A2)] for all firing sequences $(N_1, m_1) \wtrans{\varrho_1} (N_1,
    m_1')$ in $N_1$, there is at least one marking $m'_2$ over $P_2$ such
    that $m_1' \uplus m_2' \models E$ (e.g. solvable), and for all
    markings $m_2'$ over $P_2$ such that $m_1' \uplus m_2' \models E$ there must
    exist a firing sequence $\varrho_2 \in T_2^\star$ such that $(N_2, m_2)
    \wtrans{\varrho_2} (N_2, m_2')$ and $l_1(\varrho_1) = l_2(\varrho_2)$.
  \end{description}
  We say that $(N_1, m_1)$ is $E$-equivalent to $(N_2, m_2)$, denoted
  $(N_1, m_1) \reduc_E (N_2, m_2)$, when we have both
  $(N_1, m_1) \sqsupseteq_E (N_2, m_2)$ and
  $(N_2, m_2) \sqsupseteq_E (N_1, m_1)$.
\end{definition}

Notice that condition (A2) is defined only for observation sequences
starting from the initial marking of $N_1$. Hence the relation is
usually not true on every pair of matched markings; it is not a
bisimulation. Also, condition (A2) can be defined in an alternative
way using observation sequences.

\begin{description}
\item[(A2)] for all observation sequences $\sigma$ such that $(N_1, m_1)
  \wtrans{\sigma} (N_1, m_1')$, there is at least one marking $m'_2$ over $P_2$
  such that $m_1' \uplus m_2' \models E$, and for all markings $m_2'$ over $P_2$
  such that $m_1' \uplus m_2' \models E$ we must have $(N_2, m_2)
  \wtrans{\sigma} (N_2, m_2')$.
\end{description}

By definition, relation $\reduc_E$ is symmetric. We deliberately use a
``comparison symbol'' for our equivalence, $\reduc$, in order to
stress the fact that we expect the fact that $N_2$ is a reduced
version of $N_1$. In particular, we expect that $|P_2| \le |P_1|$.

\subsection{Basic properties of polyhedral abstraction}

We prove that we can use $E$-equivalence to check the reachable markings of
$N_1$ simply by looking at the reachable markings of $N_2$. We give a first
property that is useful in the context of {bounded model checking}, when we try
to find a counter-example to a property by looking at firing sequences with
increasing length. Our second property is useful for checking invariants, and is
at the basis of our implementation of the PDR method for Petri nets.

\begin{lemma}[Reachability checking]\label{reachability} Assume $(N_1, m_1)
  \reduc_E (N_2, m_2)$. Then for all $m_1'$ in $R(N_1, m_1)$ there is $m_2'$ in
  $R(N_2,m_2)$ such that $m_1' \uplus m_2' \models E$.
\end{lemma}

\begin{proof}
  Since $m_1'$ is reachable, there must be a firing sequence $\varrho_1$ in $N_1$
  such that $(N_1, m_1) \wtrans{\varrho_1} (N_1, m_1')$. By condition (A2), there
  must be some marking $m_2'$ over $P_2$, compatible with $m_1'$, such that
  $m_1' \uplus m_2' \models E$ and $(N_2, m_2) \wtrans{\varrho_2} (N_2, m_2')$
  (for some firing sequence $\varrho_2$). Therefore we have $m_2'$ reachable in
  $N_2$ such that $m_1' \uplus m_2' \models E$.
\end{proof}

Lemma~\ref{reachability} can be used to find a counter-example $m'_1$, to some
property $F$ in $N_1$, just by looking at the reachable markings of $N_2$.
Indeed, it is enough to find a marking $m_2'$ reachable in $N_2$ such that
${m_2'} \models E \wedge \neg F$. This is the result we use in our
implementation of the BMC method.

Our second property can be used to prove that every reachable marking of $N_2$
can be traced back to at least one marking of $N_1$ using the reduction
constraints. (While this mapping is surjective, it is not a function, since a
state in $N_1$ could be associated with multiple states in $N_2$.)

\begin{lemma}[Invariance checking]\label{iclemma} Assume $(N_1, m_1) \reduc_E
  (N_2, m_2)$. Then for all pairs of markings $m_1', m'_2$ of $N_1, N_2$ such
  that $m_1' \uplus m'_2 \models E$ and $m'_2 \in R(N_2,m_2)$ it is the case
  that $m_1' \in R(N_1, m_1)$.
\end{lemma}

\begin{proof}
  Take $m_1', m_2'$ a pair of markings in $N_1, N_2$ such that $m_1' \uplus m_2'
  \models E$ and $m_2' \in R(N_2,m_2)$. Hence there is a firing sequence
  $\varrho_2$ such that $(N_2, m_2) \wtrans{\varrho_2} (N_2, m_2')$. By condition
  (A2), since $m_1' \uplus m_2' \models E$, there must be a firing sequence in
  $N_1$, say $\varrho_1$, such that $(N_1, m_1) \wtrans{\varrho_1} (N_1, m_1')$.
  Hence $m_1' \in R(N_1, m_1)$.
\end{proof}

\vspace*{-1.8mm}
Using Lemma~\ref{iclemma}, we  can easily extract an invariant on $N_1$ from an
invariant on $N_2$. Basically, if property $E \wedge F$ is an invariant on $N_2$
(where $F$ is a formula whose variables are in $P_1$) then we can prove that $F$
is an invariant on $N_1$. This property (the \emph{invariant conservation}
theorem of Sect.~\ref{sec:smt-based-model}) ensures the soundness of the model
checking technique implemented in our tool.

\subsection{Composition laws}

We prove that polyhedral abstraction is a transitive relation
(Th~\ref{th:transitivity}) that is also closed by synchronous
composition (Th~\ref{th:composition}) and relabeling
(Th~\ref{th:relabeling}). These results can be used as a set of
``algebraic laws'' allowing us to derive complex equivalence
assertions from much simpler instances, or \emph{axioms}, inside
arbitrary contexts. We give an example of such reasoning in
Sect.~\ref{sec:deriv-e-equiv}.

\medskip
Before defining our composition laws, we start by describing
sufficient conditions in order to safely compose equivalence
relations. The goal here is to avoid {inconsistencies} that could
emerge if we inadvertently reuse the same variable in different
reduction constraints.

\medskip
The \emph{fresh variables} in an equivalence statement $\mathrm{EQ} : (N_1, m_1)
\reduc_E (N_2, m_2)$ are the variables occurring in $E$ but not in $P_1 \cup
P_2$. (These variables can be safely ``alpha-converted'' in $E$ without changing
any of our results.)  We say that a net $N_3$ is \emph{compatible} with respect
to $\mathrm{EQ}$ when $(P_1 \cup P_2) \cap P_3 = \emptyset$ and there are no
fresh variables of $\mathrm{EQ}$ that are also places in $P_3$. Likewise we say
that the equivalence statement $\mathrm{EQ}': (N_2, m_2) \reduc_{E'} (N_3, m_3)$
is \textit{compatible} with $\mathrm{EQ}$ when $P_1 \cap P_3 \subseteq P_2$ and
the fresh variables of $\mathrm{EQ}$ and $\mathrm{EQ}'$ are disjoint.

The composition laws stated in the following theorems are useful to build larger
equivalences from simpler axioms (reduction rules). We show some examples of
reductions in the next section and how they occur in the example of
Fig.~\ref{fig:stahl}.

\subsubsection{Preservation by chaining}

We prove that we can chain equivalences together in order to derive
more general reduction rules. When doing so, we need to combine
constraints together.

\begin{theorem}\label{th:transitivity}
	Assume we have two compatible equivalence statements
	$(N_1, m_1) \reduc_E (N_2, m_2)$ and
	$(N_2, m_2) \reduc_{E'} (N_3, m_3)$, then
	$(N_1, m_1) \reduc_{E \land E'} (N_3, m_3)$.
\end{theorem}

\begin{proof}
	First, we use the fact system $E \land E'$ is solvable for $N_1, N_3$. This is
	a consequence of the compatibility assumption, since no fresh variable in $E$
	can clash with a fresh variable in $E'$. For similar reason, we have that $m_1
	\uplus m_2 \models E$ and $m_2 \uplus m_3 \models E'$ entails $m_1 \uplus m_3
	\models E \land E'$. Indeed we even have the stronger property that
	$\underline{m_1} \land \underline{m_2} \land \underline{m_3} \land  E \land
	E'$ is satisfiable. From this, we obtain condition (A1) and the first
	constraint of condition (A2).
	
	For the second constraint of condition (A2), we assume that $\sigma$ is an
	observation sequence such that $(N_1, m_1) \wtrans{\sigma} (N_1, m_1')$.
	Hence, using the fact that $(N_1, m_1) \reduc_E (N_2, m_2)$, we have $(N_2,
	m_2) \wtrans{\sigma} (N_2, m_2')$ for every marking $m_2'$ of $N_2$ such that
	$m_1' \uplus m_2' \models E$. Using a similar property from $(N_2, m_2)
	\reduc_{E'} (N_3, m_3)$, we have $(N_3, m_3) \wtrans{\sigma} (N_3, m_3')$ for
	every marking $m_3'$ of $N_3$ such that $m_2' \uplus m_3' \models E$. The
	result follows from the observation that, since $E$ and $E'$ are both solvable
	and the nets are compatible, for all markings $m_1''$ of $N_1$, if a marking
	$m_3''$ of $N_3$ satisfies $m_1'' \uplus m_3'' \models E \land E'$ then there
	must be a marking $m_2''$ of $N_2$ such that both $m_1'' \uplus m_2'' \models
	E$ and $m_2'' \uplus m_3'' \models E'$.
\end{proof}

\subsubsection{Preservation by synchronous composition}

Our next result relies on the classical synchronous product operation
between labeled Petri nets \cite{lloret_compositional_1991}. Assume
$N_1 = (P_1, T_1, \pre_1, \post_1)$ and
$N_2 = (P_2, T_2, \pre_2, \post_2)$ are two labeled
Petri nets with respective labeling functions $l_1$ and $l_2$ on the
respective alphabets $\Sigma_1$ and $\Sigma_2$. We can assume, without
loss of generality, that the sets $P_1$ and $P_2$ are disjoint.

\medskip
We introduce a new symbol, $\circ$, used to build (structured) names
for transitions that are not synchronized. The \textit{synchronous
  product} between $N_1$ and $N_2$, denoted as $N_1 \| N_2$, is the
net $(P_1 \cup P_2, T, \pre, \post)$ with labelling
function $l$ where $T$ is the smallest set containing:
\begin{itemize}
\itemsep=0.9pt
\item transition $(t, \circ)$ if $l_1(t) \not\in \Sigma_2$, such that $l((t, \circ)) = l_1(t)$;
\item transition $(\circ, t)$ if $l_2(t) \not\in \Sigma_1$, such
  that $l((\circ, t)) = l_2(t)$;
\item and transition $(t_1, t_2)$ if $l_1(t_1) = l_2(t_2)$, such that
  $l((t_1, t_2)) = l_1(t_1)$.
\end{itemize}

The flow functions of $N_1 \| N_2$ are such that
$\pre((t_1, t_2), p) = \pre_1(t_1, p)$ if $p \in P_1$
and $t_1 \neq \circ$, or $\pre_2(t_2, p)$ if $p \in P_2$ and
$t_2 \neq \circ$ (and $0$ in all the other cases). Similarly for
$\post$.\\

To simplify our proofs, we define a notion of projection over firing
sequences of $N_1 \| N_2$, that is two functions $\varrho \cdot 1$ and
$\varrho \cdot 2$ such that $\epsilon \cdot i = \epsilon$ and
$\left ( \varrho \, t \right ) \cdot i = \left ( \varrho \cdot i \right
) \, \left ( t \cdot i \right ) \,$ for all $i \in 1..2$, where
$(t_1, \circ) \cdot 1 = t$, and $(\circ, t_2) \cdot 1 = \epsilon$, and
$(t_1, t_2) \cdot 1 = t_1$ (and symmetrically with $\cdot 2$ on the
second component of each transition pair).

Projections can be used to extract from a firing sequence of
$N_1 \| N_2$, the transitions that were fired from the left
($\cdot 1$) and right ($\cdot 2$) components of the synchronous
product.

\medskip
We also need to define a dual relation, denoted
$\varrho_1 \| \varrho_2$, that defines the (potential) ``zip merge'' of
firing sequences in $T_1^\star \times T_2^\star$ into firing sequences
of $N_1 \| N_2$, when the two sequences can synchronize. When defined,
$\varrho_1 \| \varrho_2$ is the smallest set of sequences of
$N_1 \| N_2$ satisfying the following inductive rules. In particular,
we say that $\varrho_1$ and $\varrho_2$ can be synchronized when
$\varrho_1 \| \varrho_2 \neq \emptyset$.
\begin{itemize}
\item $\epsilon \| \epsilon = \{ \epsilon \}$
\item
  $(t_1 \, \varrho_1) \| \epsilon
  = \left \{ \begin{array}[c]{ll}
               \displaystyle \{ (t_1, \circ) \, \varrho \mid \varrho \in
               (\varrho_1 \| \epsilon) \} & \text {if all the transitions in
                                           $t_1 \, \varrho_1$ have labels in $\Sigma_1 \setminus
                                           \Sigma_2$,}\\
               \emptyset &
                           \text{otherwise.}\\
             \end{array} \right .$

\item
  $\epsilon \| (t_2 \, \varrho_2)
  = \left \{ \begin{array}[c]{ll} \displaystyle \{ (\circ, t_2) \,
               \varrho \mid \varrho \in (\epsilon \| \varrho_2) \}
               & \text {if all
                 the transitions in $t_2 \, \varrho_2$ have labels in $\Sigma_2
                 \setminus
                 \Sigma_1$,}\\
               \emptyset & \text{otherwise.}\\
             \end{array} \right .$

\item
  $(t_1\, \varrho_1) \| (t_2\, \varrho_2)
  = \left \{ \begin{array}[c]{ll}

               \displaystyle \{ (t_1, t_2) \,
               \varrho \mid \varrho \in (\varrho_1 \| \varrho_2) \}
               & \text {if $l_1(t_1) = l_2(t_2)$,}\\

               \{ (t_1, \circ) \, \varrho \mid  \varrho \in
              \varrho_1 \| (t_2\, \varrho_2) \}
               & \text {if $l_1(t_1) \in \Sigma_1
                 \setminus \Sigma_2$,}\\

               \{ (\circ, t_2) \, \varrho \mid  \varrho \in
               (t_1\, \varrho_1) \| \varrho_2 \}
               & \text {if $l_2(t_2) \in
                 \Sigma_2 \setminus \Sigma_1$,}\\

               \emptyset &
                           \text{otherwise.}\\
                                            \end{array} \right .$
\end{itemize}

We can also project the reachable markings of a synchronous product
over reachable markings of each of its components. Since the places in
$N_1$ and $N_2$ are disjoint, we can always see a marking $m$ in $N_1
\|
N_2$ as the disjoint union of two (necessarily compatible) markings
$m_1, m_2$ from $N_1, N_2$. In this case we simply write $m = m_1 \|
m_2$.

\medskip
More generally, we extend this product operation to marked nets and
write $(N_1, m_1) \| (N_2, m_2)$ for the marked net
$(N_1 \| N_2, m_1 \| m_2)$. The following result underscores the
equivalence between the semantics (the Labeled Transition System) of
$N_1 \| N_2$ and the product of the LTS of its components.

\begin{lemma}[Projection and product of sequences]\label{lemma:projection}
  Assume there is a firing sequence
  $(N_1 \| N_2, m_1 \| m_2) \wtrans{\varrho} (N_1 \| N_2, m_1' \|
  m_2')$ on the synchronous product $N_1 \| N_2$. Then the projections
  $\varrho \cdot 1$ and $\varrho \cdot 2$ are firing sequences of their
  respective components,
  $(N_i, m_i) \wtrans{\varrho \cdot i} (N_i, m_i')$ for all
  $i \in 1..2$, such that $\varrho \cdot 1$ and $\varrho \cdot 2$ can be
  synchronized: $\varrho \cdot 1 \| \varrho \cdot 2 \neq \emptyset$.
  Conversely, if $(N_i, m_i) \wtrans{\varrho_i} (N_i, m_i')$ for all
  $i \in 1..2$ and $\varrho \in (\varrho_1 \| \varrho_2)$ then
  $(N_1 \| N_2, m_1 \| m_2) \wtrans{\varrho} (N_1 \| N_2, m_1' \|
  m_2')$.
\end{lemma}

\begin{proof}
  See for instance Proposition~$2.1$
  in~\cite{lloret_compositional_1991}.
\end{proof}

We can now prove that $E$-abstraction equivalence is stable by
synchronous composition. Note that it is enough to prove the results
on $E$-abstraction, since the equivalence is symmetric.

\begin{theorem}[Composability]\label{th:composition}
  Assume $(N_1, m_1) \reduc_E (N_2, m_2)$ and that $(M, m)$ is
  {compatible} with respect to this equivalence, then
  $(N_1, m_1) \| (M, m) \reduc_E (N_2, m_2) \| (M, m)$.
\end{theorem}

\begin{proof}
  By hypothesis system $E$ is solvable for $N_1, N_2$. Hence, since
  $M$ is compatible, no place in the net $M$ can occur in one of the
  constraints of $E$. Therefore $E$ is also solvable for the pair of
  nets $(N_1 \| M)$ and $(N_2 \| M)$.  Likewise, the initial markings
  $(m_1 \| m)$ and $(m_2 \| m)$ are compatible together and
  $(m_1 \| m) \uplus (m_2 \| m) \models E$ (the constraints in $m$
  have no effect on the constraints of $E$). Therefore condition (A1) is
  valid for the marked nets $(N_1, m_1) \| (M, m)$ and
  $(N_2, m_2) \| (M, m)$, and we obtain the first constraint of condition (A2).
	
  We are left with proving the second constraint of condition (A2).
  Assume we have a firing sequence $\varrho$ in $N_1 \| M$. By our
  projection property (Lemma~\ref{lemma:projection}) it must be the
  case that
  $(N_1 \| M, m_1 \| m) \wtrans{\varrho} (N_1 \| M, m_1' \| m')$ with
  $(N_1, m_1) \wtrans{\varrho \cdot 1} (N_1, m_1')$. We also have that
  $(M, m) \wtrans{\varrho \cdot 2} (M, m')$ such that
  $(\varrho \cdot 1) \| (\varrho \cdot 2) \neq \emptyset$.
	
  By condition (A2) on the abstraction between $N_1$ and $N_2$, it
  must be the case that $(N_2, m_2) \wtrans{\varrho_2} (N_2, m_2')$,
  for some firing sequence $\varrho_2$ of $N_2$, for all markings $m_2'$ of
  $N_2$ such that $m_1'\uplus m_2'\models E$. Moreover the observable
  sequence obtained from $\varrho_2$ and $\varrho \cdot 1$ are the same:
  $l_1(\varrho \cdot 1) = l_2(\varrho_2)\ (\star)$, which means also
  that $(\varrho_2) \| (\varrho \cdot 2) \neq \emptyset$. Hence,
  using the second direction in Lemma~\ref{lemma:projection}, we can
  find a firing sequence in $\varrho_2 \| (\varrho \cdot 2)$, say
  $\varrho'$, such that
  $(N_2 \| N_3, m_2 \| m_3) \wtrans{\varrho'} (N_2 \| M, m_2' \|
  m')$. Like in the proof of condition (A1), we obtain that
  $(m_1' \| m') \uplus (m_2' \| m') \models E$ from the fact that
  $m_1'\uplus m_2'\models E$, and $E$ is solvable, and $M$ is
  compatible.
	
  We are left to prove that $\varrho$ and $\varrho'$ have the same
  observation sequences. This is a consequence of the fact that
  $l_1(\varrho \cdot 1) = l_2(\varrho_2)$ (property $\star$ above); and
  the fact that, by construction of $\varrho'$, we have
  $\varrho' \cdot 1 = \varrho_2$ and $\varrho' \cdot 2 = \varrho \cdot 2$.
\end{proof}

\subsubsection{Preservation by relabeling}

Another standard operation on labeled Petri nets is \emph{relabeling},
denoted as $N[a/b]$, that apply a substitution to the labeling
function of a net. Assume $l$ is the labeling function over the
alphabet $\Sigma$. We denote $l[a/b]$ the labeling function on
$(\Sigma \setminus \{ a \}) \cup \{ b\}$ such that $l[a/b](t) = b$
when $l(t) = a$ and $l[a/b](t) = l(t)$ otherwise. Then $N[a/b]$ is the
same as net $N$ but equipped with labeling function
$l[a/b]$. Relabeling has no effect on the marking of a net.  The
relabeling law is true even in the case where $b$ is the silent action
$\tau$. In this case we say that we \emph{hide} action $a$ from the
net.

We prove that $E$-abstraction equivalence is also preserved by
relabeling and hiding.

\begin{theorem}\label{th:relabeling}
  If $(N_1, m_1) \reduc_E (N_2, m_2)$ then
  $(N_1[a/b], m_1) \vartriangleright_{E} (N_2[a/b], m_2)$.
\end{theorem}

\begin{proof}
	Assume $(N_1, m_1) \reduc_E (N_2, m_2)$. Condition (A1)
	does not depend on the labels and therefore it is also true between
	$N_1[a/b], E$ and $N_2[a/b]$. For condition (A2), we simply use the fact that for any firing sequences $\varrho_1$ and $\varrho_2$, $l_1(\varrho_1) = l_2(\varrho_2)$ implies $l_1[a/b](\varrho_1) = l_2[a/b](\varrho_2)$.
\end{proof}

\subsection{Reduction rules}

We define a simplified set of relations that can act as ``axioms'' in
a system for deriving $E$-abstraction equivalences. Each of these
axioms derives from a standard \emph{structural reduction rule} (see
e.g~\cite{berthelot_transformations_1987,berthomieu_counting_2019}),
where labeled transitions play the role of interfaces with a possible
outside ``context''.

Each rule is defined by a triplet $(N_1, E, N_2)$ such that
$(N_1, m_1) \reduc_E (N_2, m_2)$. A rule also defines possible values
for the initial markings, which can be expressed using integer
parameters, and may also include a condition that should be true
initially.

Each of our rules corresponds to instances of the reduction system that
was defined in our previous work on ``counting reachable
markings''~\cite{berthomieu_counting_2019} (we give a precise
reference in each case). Hence they also correspond to instances of
reduction rules implemented in our tool, called \Reduce, that can
automatically find occurrences of reductions in Petri nets and apply
them recursively. We give more information about this tool and the
relation with our approach in Sect.~\ref{sec:deriv-e-equiv}. This
section also contains an example showing how to apply our reduction
rules to derive the equivalence stated  in Fig.~\ref{fig:stahl}.

We consider four general families of reductions: first rules for
agglomerating places (like [CONCAT] and [AGG]); then rules based on a
``place invariant'' over the initial net (what we call a
\emph{redundancy rule} like [RED] and [SHORTCUT]); rules for garbage collecting dead places or
transitions (like [DEADT] and [REDT]); and finally rules that can be used to abstract constant
or ``closed'' places (like [CONSTANT] and [SOURCE]).

 \medskip
We give a detailed proof of correctness for our first ``reduction
axiom'', rule [CONCAT], since it is representative of the complexity
of checking simple instances of $E$-abstraction equivalence. We do not
prove similar results for all the rules defined in this section but
will only give one other example, for the redundancy rule [RED]. All
the correctness proofs for the reduction rules given in this section
are very similar to one of these two examples.

\subsubsection{Rule [CONCAT]}
Our first example is the prototypical example of net reduction, as
defined in~\cite{berthelot_transformations_1987}. It also corresponds
to the simplest example of agglomeration rule (see the rule for
chaining in Fig.~$6$ of~\cite{berthomieu_counting_2019}).

\medskip
Rule [CONCAT], Fig.~\ref{fig:concat}, can be used to fuse together two
places ``connected only through a deterministic transition'' (modeled
as a silent transition in our approach). The constraints imposed for
applying this rule is that place $y_2$, in the initial net, must be
initially empty. We also have the condition that no transition with
an observable label (hence no transition that can potentially be
merged with an outside context with a synchronous composition) can add
a token directly to place $y_2$. This condition is necessary to ensure
the correctness of this rule, see Proposition~\ref{prop:rule-concat}.

\begin{figure}[tb]
  \centering
  \includegraphics{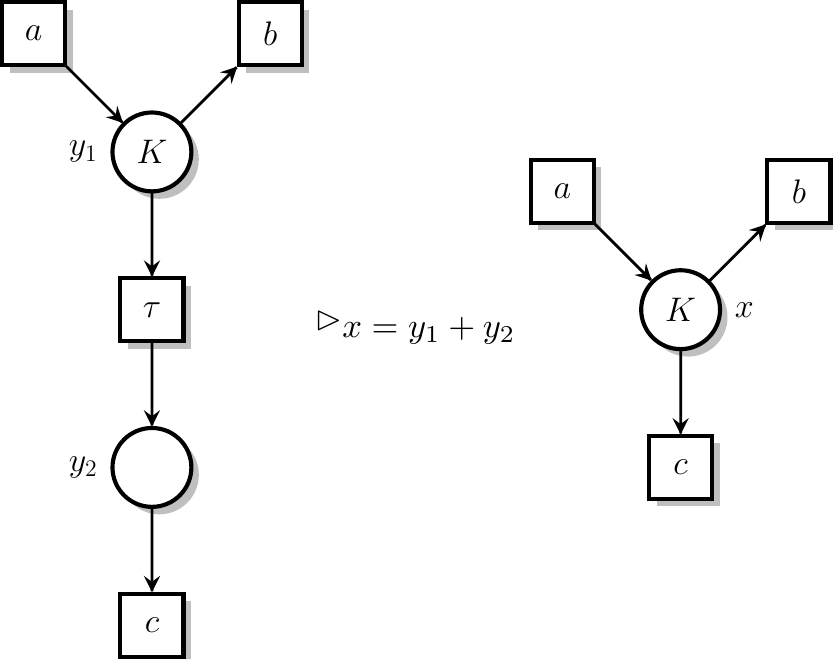}
  \caption{Rule [CONCAT].\label{fig:concat}}\vspace*{-3mm}
\end{figure}

\medskip
Note that nets $N_1$ and $N_2$ are not bounded, since transition $a$
can always be fired to increase the marking of places $y_1$ and
$x$. Which means that we need to consider an unbounded number of
firing sequences. Therefore it may not be possible to prove this
result using an automatic verification method, such as
model checking. Actually, we are working on the definition of an
automatic proof method for certifying the correctness of reduction
rules, which would be an interesting addition to our approach.

\begin{proposition}[Correctness of rule {[CONCAT]}]\label{prop:rule-concat}
  We have $(N_1, m_1) \reduc_E (N_2, m_2)$, with $E$ the system
  containing the single equation $x = y_1 + y_2$, and $N_1, N_2$ the
  nets depicted in Fig.~\ref{fig:concat}.
 \end{proposition}

\eject
\begin{proof}
  The constraints on the initial marking of the nets are such that
  $m_1(y_1) = m_2(x) = K \geq 0$ and $m_1(y_2) = 0$. To ease the
  presentation, we should use $\tau, a, b, c$ as the name of the
  transitions, and not only as labels.

  We start by proving condition (A1) and the first constraint of condition (A2).
  By construction, we have $m_1 \uplus m_2 \models E$ and $E$ is solvable for
  $N_1, N_2$. Indeed, equation $x = y_1 + y_2$ is always satisfiable when we fix
  either the values of variables $y_1, y_2$, or the value of $x$.

  We now prove the second constraint of condition (A2) for the relation
  $(N_1, m_1) \sqsupseteq_E (N_2, m_2)$. Assume that
  $(N_1, m_1) \wtrans{\varrho_1} (N_1, m_1')$ and that
  $m_1' \uplus m_2' \models E$. By definition of $E$, when $m_1'$ is
  fixed, there is a unique solution for $m_2'$ such that
  $m_1' \uplus m_2' \models E$; which is
  $m_2'(x) = m_1'(y_1) + m_1'(y_2)$. Take $\varrho_2$ the unique firing
  sequence of $N_2$ such that  $l_1(\varrho_1) = l_2(\varrho_2)$
  (basically $\varrho_2$ is obtained from $\varrho_1$ by erasing all
  occurrences of the silent transition). It is the case that
  $(N_2, m_2) \wtrans{\varrho_2} (N_2, m_2')$, as needed.

  We are left to prove the same constraint for the relation
  $(N_2, m_2) \sqsupseteq_E (N_1, m_1)$. Assume we have
  $(N_2, m_2) \wtrans{\varrho_2} (N_2, m_2')$. We prove that there is a
  firing sequence $\varrho_1$ such that
  $(N_1, m_1) \wtrans{\varrho_1} (N_1, m_1')$ and
  $l_1(\varrho_1) = l_2(\varrho_2)$, where $m_1'$ is the marking defined
  by $m_1'(y_1) = m_2'(x)$ and $m_1'(y_2) = 0$ (all the tokens are in
  $y_1$). We define $\varrho_1$ as the (unique) sequence obtained from
  $\varrho_2$ by adding one occurrence of the $\tau$-transition before
  each occurrence of $c$ in $\varrho_2$. Intuitively, we always keep
  all the tokens in place $y_1$ of $N_1$, except before firing a $c$;
  in which case we add a token to place $y_2$. We can prove, using an
  induction on the size of $\varrho_2$, that $\varrho_1$ is a legitimate
  firing sequence of $(N_1, m_1)$ and that
  $(N_1, m_1) \wtrans{\varrho_1} (N_1, m_1')$.
\end{proof}

\subsubsection{Rule [AGG]}

Our second example of rule is for the \emph{agglomeration of places}, see
Fig.~\ref{fig:agg}, that can be used to simplify a ``cluster of places'' between
which tokens can move freely from $y_2$ to $y_1$. This is an instance of the
general ``loop agglomeration'' rule given in Fig.~$7$
of~\cite{berthomieu_counting_2019}.

We could easily define a family of reduction rules similar to [AGG]
and [CONCAT] but for longer ``loops'' or ``chains'' of places, or with
the addition of weights on the arcs. For the sake of brevity, we only
list one archetypal instance of each rule in this section.

\begin{figure}[tb]
  \centering \includegraphics{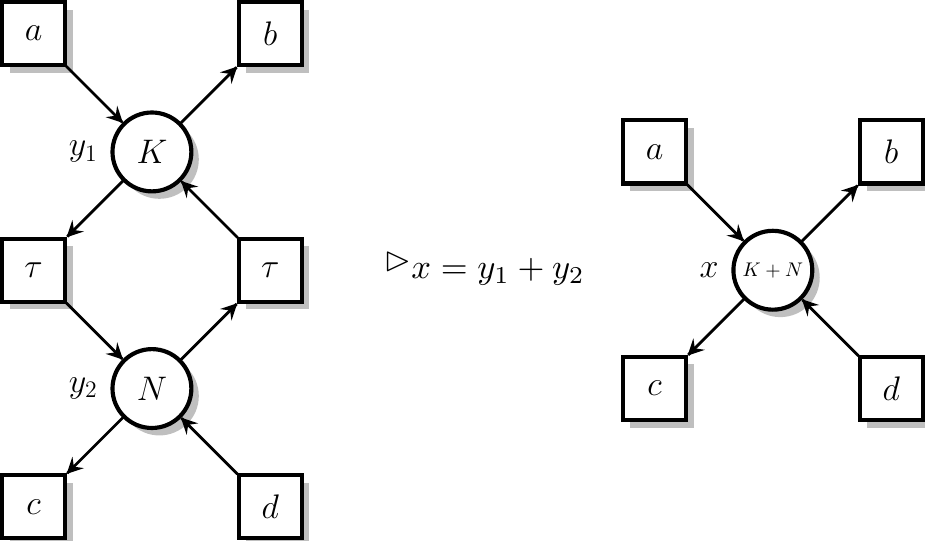}
  \caption{Rule [AGG].\label{fig:agg}}\vspace*{-3mm}
\end{figure}

\subsubsection{Rules [RED] and [SHORTCUT]}

Our next two rules, Fig.~\ref{fig:red}, are reductions that can be used to
eliminate \emph{redundant places}, meaning places whose marking derives from a
place invariant (and the knowledge of the marking of other places).
\begin{figure}[!h]
  \centering \includegraphics{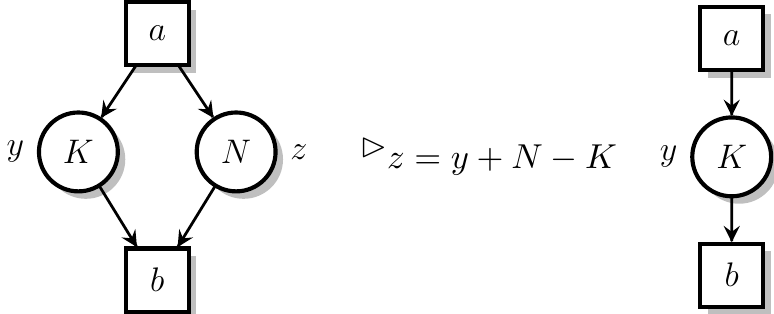}

  \vspace{1ex} \hrulefill

  \vspace{3ex}

  \includegraphics[width=0.95\textwidth]{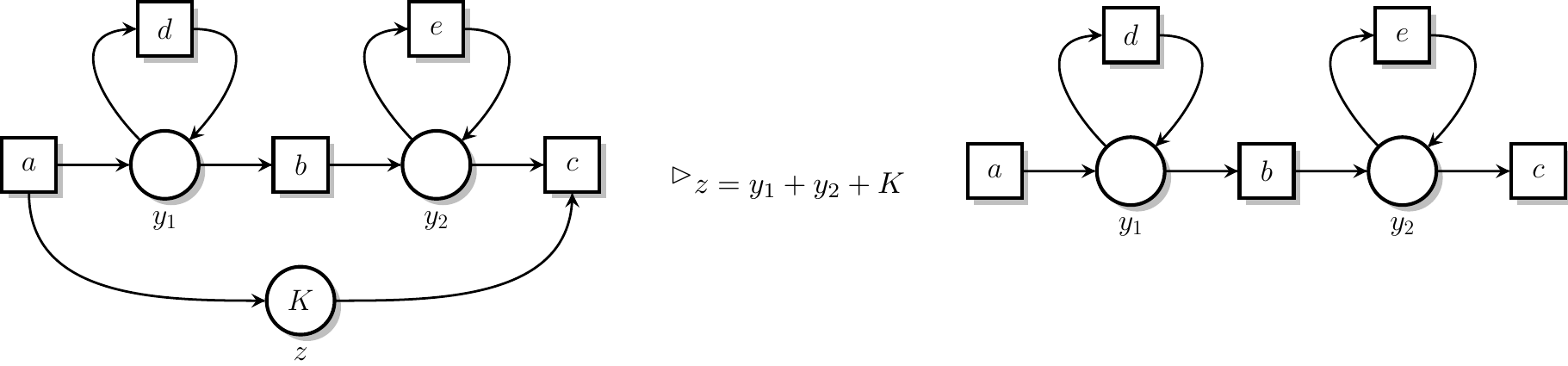}

  \caption{Rule [RED] (above), assuming $K \leqslant N$, and rule
    [SHORTCUT] (below).\label{fig:red}}\vspace*{-3mm}
\end{figure}

In rule [RED] for instance, with the assumption that we have more
tokens in place $z$ than in $y$ initially, it is always the case that
$m(z) - m(y)$ is a constant for all the reachable states $m$. Hence we
can safely eliminate $z$ and keep the relevant information in our
linear system $E$.

Rule [SHORTCUT] gives a more involved example, that relies on a
condition involving more than two places; an invariant of the form $z =
  y_1 + y_2 + K$.

We give the proof of correctness for the equivalence corresponding to
rule [RED]. The proofs for other redundant place elimination rules are
all similar.

\begin{proposition}[Correctness of rule {[RED]}]\label{prop:rule-red}
  Assuming $K \leqslant N$, we have $(N_1, m_1) \reduc_E (N_2, m_2)$,
  with $E$ the system containing the single equation $x = y + N - K$,
  and $N_1, N_2$ the nets depicted in Fig.~\ref{fig:red} (above).
\end{proposition}

\begin{proof}
  Condition $K \leqslant N$ is necessary in order to have that
  $N - K \geqslant 0$, and therefore that the marking of $y$ in $N_2$
  is indeed non-negative.

  By construction, we have $m_1 \uplus m_2 \models E$, satisfying condition
  (A1). The first constraint of condition (A2) follows from the fact that $z =
  y + N - K$ is an invariant on $(N_1, m_1)$, meaning that for all firing
  sequence $\varrho$ such that $(N_1, m_1) \wtrans{\varrho} (N_1, m_1')$ we have
  $m_1'(z) = m_1'(y) + N - K$. This can be proved by a simple induction on the
  length of $\varrho$. Hence, $E$ is satisfied for every reachable marking in
  $(N_1, m_1)$, and so $E$ is solvable.

  We now prove the second constraint of condition (A2) for the relation $(N_1,
  m_1) \sqsupseteq_E (N_2, m_2)$. Assume that $(N_1, m_1) \wtrans{\varrho} (N_1,
  m_1')$. We have that $\varrho$ is also a firing sequence of $(N_2, m_2)$ and,
  moreover, $(N_2, m_2) \wtrans{\varrho} (N_2, m_2')$ such that $m_2'(y) =
  m_1'(y)$. The proof is similar in the other direction.
\end{proof}

\subsubsection{Rules [REDT] and [DEADT]}

We can use the same approach to simplify transitions in a net, rather than
places. One such example is rule [REDT], to remove redundant transitions. Such
rules are interesting because, when applied in collaboration with others, they
can create new opportunities to apply reductions. We give an example of such
mechanism in the example of Sect.~\ref{sec:deriv-e-equiv}.

Another example is the elimination of dead transitions, rule [DEADT],
that can get rid of transitions that are ``structurally dead''. In
this example, we know that place $x$ will always stay empty since no
transition can increase its marking. Hence the $\tau$ transition is
dead and we can remove it without modifying the set of reachable
markings nor the observation sequences.

\begin{figure}[tb]
  \centering \includegraphics{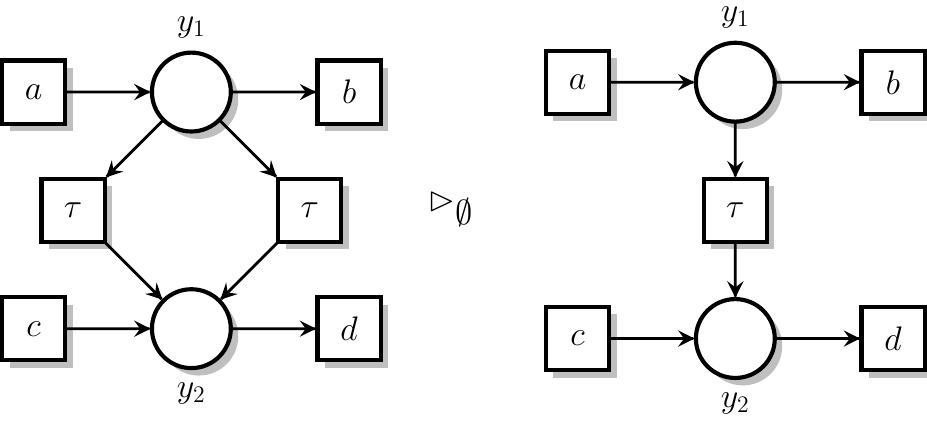}

  \vspace{1ex} \hrulefill

  \vspace{3ex}

  \includegraphics{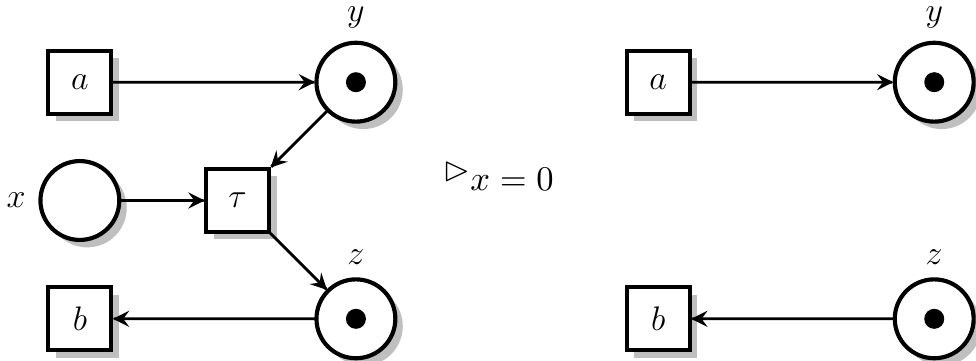}

  \caption{Rules [REDT] (above), and
    [DEADT] (below).\label{fig:redt}}
\end{figure}

\subsubsection{Rules [CONSTANT] and [SOURCE]}

Our last examples of rules illustrate the case of equivalences $(N_1, m_1)
\reduc_{E} (N_2, m_2)$ where the final Petri net is ``empty'' (denoted
$\emptyset$). A Petri net with an empty set of places has only one marking; the
empty mapping, the only function with domain $\emptyset \to \Nat$.

In this case the reachable markings of $(N_1, m_1)$ are exactly
defined by the non-negative solutions of system $E$.

Such cases may occur in practice when we can apply several reductions
in a row. We say that the initial net is ``fully reducible''. In
example [SOURCE] for instance, we can abstract the state space of the
initial net with the single constraint $x \leqslant K$.

\begin{figure}[tb]
\vspace*{3mm}
  \hfill \includegraphics{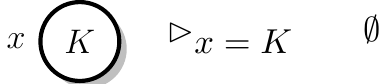}
  \hfill
  \includegraphics{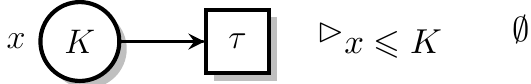}
  \hfill \hspace*{0em}
  \caption{Rules [CONSTANT] (left)  and
    [SOURCE] (right).\label{fig:constant}}
\end{figure}

We have other rules that allow us to fully reduce a net. For instance
specific structural or behavioural restrictions, such as nets that are
marked graphs or other cases where the set of reachable markings is
exactly defined by the solutions of the state
equation~\cite{hujsa2020checking}.

\subsection{Deriving $E$-equivalences using reductions}
\label{sec:deriv-e-equiv}

We can compute net reductions by reusing a tool, called \Reduce, that
was developed in our previous work~\cite{berthomieu_counting_2019}
(see also Sect.~\ref{sec:experimental-results}). The tool takes a
marked Petri net as input and returns a reduced net and a sequence of
linear constraints. For example, given the net $M_1$ of
Fig.~\ref{fig:stahl}, \Reduce returns net $M_2$ and equations
$(p_5 = p_4), (a_1 = p_1 + p_2), (a_2 = p_3 + p_4)$, and
$(a_1 = a_2)$, that corresponds to formula $E_M$ in
Fig.~\ref{fig:stahl}.

\begin{figure}[!ht]
\vspace*{4mm}
  \centering
  \includegraphics[height=9.2em]{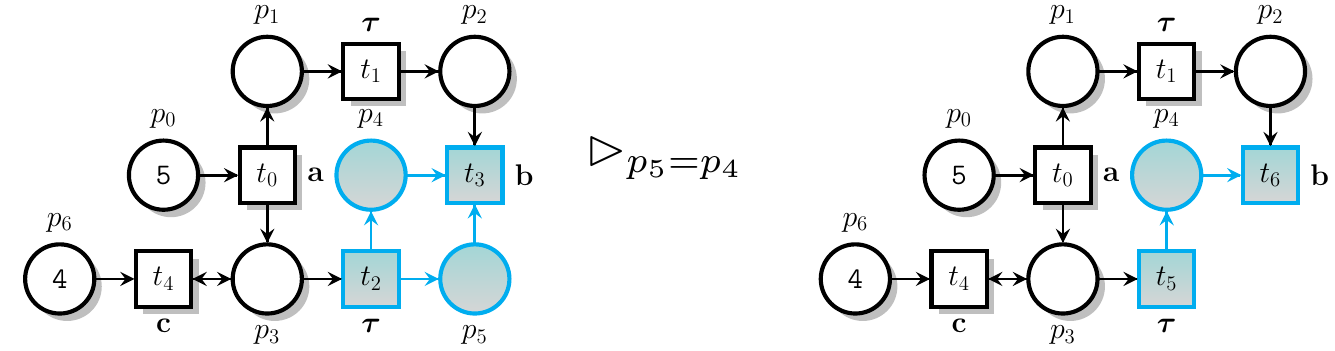}
   \vskip 0.5em
  \includegraphics[height=9.2em]{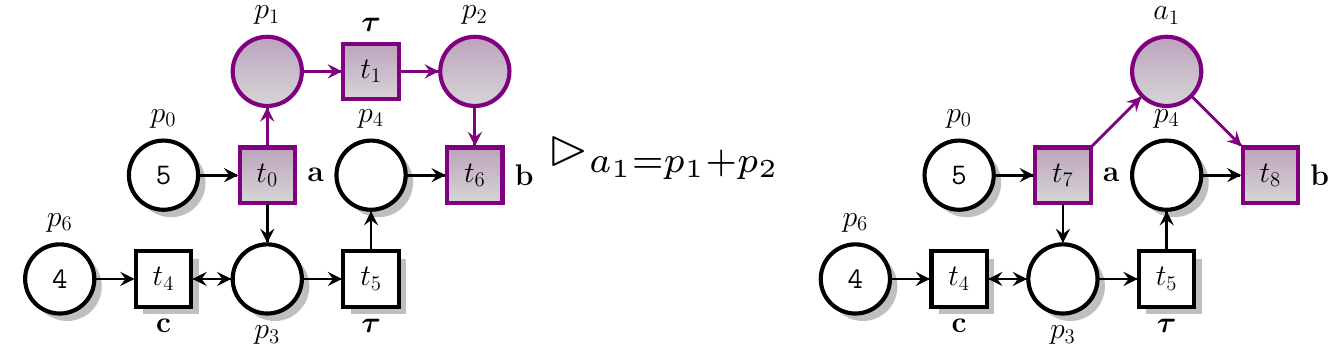}
  \vskip 0.5em
  \includegraphics[height=9.2em]{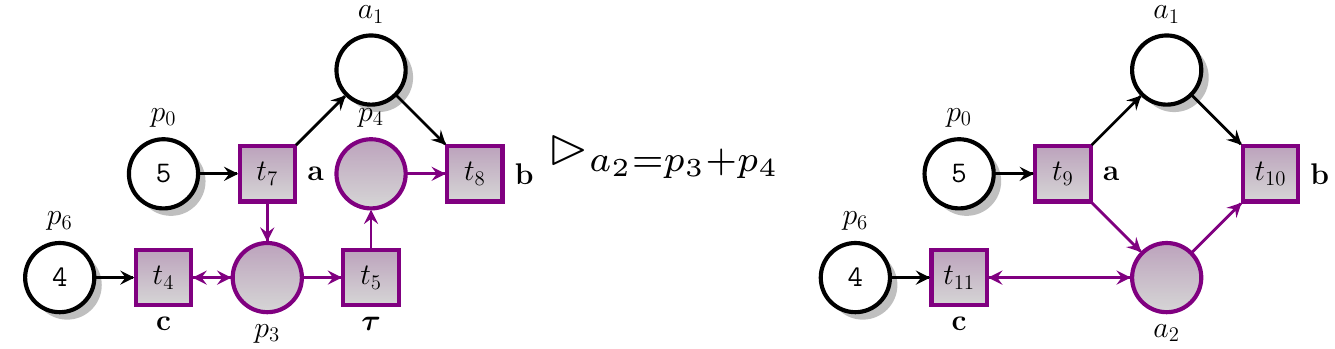}
  \vskip 0.5em
  \includegraphics[height=9.2em]{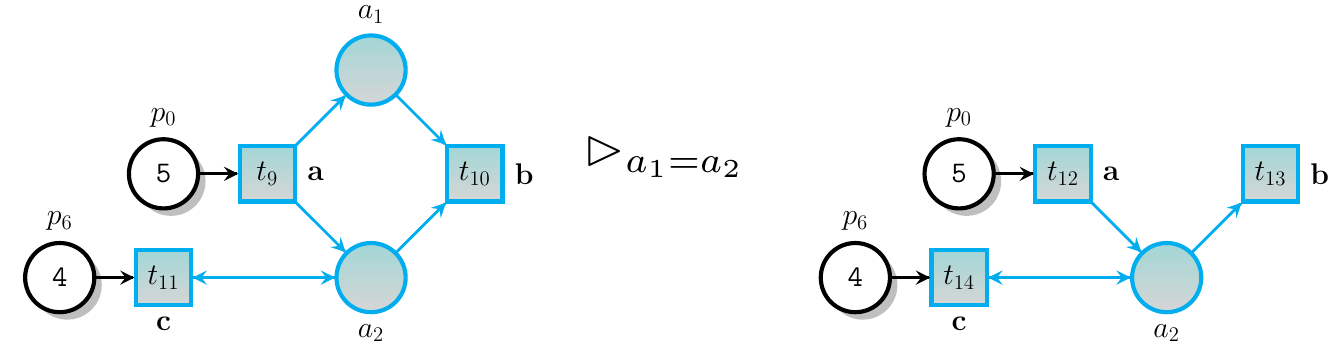}
  \caption{Example of sequence of four reductions leading from the net
    $N_1$ to $N_2$ from Fig.~\ref{fig:stahl}.\label{fig:reducs}}
\end{figure}

The tool works by applying successive reduction rules, in a
compositional way. We give an example of this mechanism in
Fig.~\ref{fig:reducs}, where we show the four reduction steps involved
in our running example.

The first step is a direct application of rule [RED] inside a larger
context; in each case we use colours to emphasize the sub-net where
the rule is applied. The two following ones are variations of rule
[CONCAT]. Each rule introducing a fresh ``place variable'', $a_1$ and
$a_2$. Finally, after simplification, we obtain a net with a new
opportunity to apply a redundancy rule.

It is possible to prove that each reduction step computed by \Reduce, from a net
$(M_i, m_i)$ to $(M_{i+1}, m_{i+1})$ with constraints $E_i$, is such that $(M_i,
m_{i}) \reduc_{E_i} (M_{i+1}, m_{i+1})$. From Th.~\ref{th:transitivity}, we
have $(M_0, m_0) \reduc_{E} (M_{n}, m_{n})$, i.e., the results computed by
\Reduce always translate into valid polyhedral abstractions.

\medskip
In conclusion, we can use \Reduce to compute polyhedral abstractions
automatically. In the other direction, we can use our notion of
equivalence to prove the correctness of new reduction patterns that
could be added in the tool. While it is not always possible to reduce
the complexity of a net using this approach, we observed in our
experiments (Sect.~\ref{sec:experimental-results}) that, on a
benchmark suite that includes almost $1\,000$ instances of nets, about
half of them
can be reduced by a factor of more than $30\%$.


\section{SMT-based model checking using abstractions}
\label{sec:smt-based-model}

We introduce a general method for combining polyhedral
abstractions with  SMT-based model checking procedures.
Assume we have $(N_1, m_1) \reduc_E (N_2, m_2)$, where the nets
$N_1, N_2$ have sets of places $P_1, P_2$ respectively. In the
following, we use $\vec{p_1} \defeq (p^{1}_{1}, \dotsc, p^{1}_{k})$ and
$\vec{p_2} \defeq (p^{2}_{1}, \dotsc, p^{2}_{l})$ for the places in $P_1$
and $P_2$. We also consider (disjoint) sequences of variables,
$\vec{x}$ and
$\vec{y}$, ranging over (the places of) $N_1$
and $N_2$.
With these notations, we denote $\tilde E(\vec{x}, \vec{y})$ the
formula obtained from $E$ where place names in $N_1$ are replaced with
variables in $\vec{x}$, and place names in $N_2$ are replaced with
variables in $\vec{y}$. When we have the same place in both nets, say
$p^{1}_{i} = p^{2}_{j}$, we also add the constraint $(x_i = y_j)$ to
$\tilde{E}$ in order to avoid shadowing variables. (Remark that
$\tilde{E}(\vec{p_1}, \vec{p_2})$ is equivalent to $E$, since
equalities $x_i = y_j$ become tautologies in this case.)
\begin{align}
  \tilde{E}(\vec{x}, \vec{y}) &\defeq E\{\vec{p_1} \leftarrow \vec{x}\}
                                \{\vec{p_2} \leftarrow \vec{y}\}\land
                                \bigwedge_{\{ (i,j) \mid p^{1}_{i} = p^{2}_{j}\}} (x_i = y_j) \label{E_tilde}
\end{align}

Given a formula $F$, we denote $\fv{F}$ the set of free variables
contained in it. Assume $F_1$ is a property that we want to study on
$N_1$, without loss of generality we can enforce the condition
$(\fv{F_1} \setminus P_1) \cap (\fv{E} \setminus P_1) = \emptyset$
(meaning we can always rename the variables in $F_1$ and $E$ that are
not places in $N_1$).  This condition ensures that the property
checked on the initial net does not inadvertently contain new variables
introduced during the reduction.

\begin{definition}[$E$-transform formula]
  Assume $(N_1, m_1) \reduc_E (N_2, m_2)$ and take $F_1$ a
  property with variables in $P_1$ such that
  $(\fv{F_1} \setminus P_1) \cap (\fv{E} \setminus P_1) =
  \emptyset$. Formula
  $F_2(\vec{y}) \defeq \exists \vec{x}. \ \tilde{E}(\vec{x}, \vec{y}) \land F_1(\vec{x})$
  is the $E$-transform of $F_1$.
\end{definition}

The following property states that, to check an invariant $F_1$ on the reachable markings of $N_1$,
it is enough to check the corresponding $E$-transform formula $F_2$ on the reachable markings of $N_2$.

\begin{theorem}[Invariant conservation]\label{th:invariant_conservation}
  Assume $(N_1, m_1) \reduc_E (N_2, m_2)$ and that
  $F_2$ is the $E$-transform of formula $F_1$ on $N_1$. Then
   $F_1$ is an invariant on $N_1$ if and only if
  $F_2$ is an invariant on $N_2$.
\end{theorem}

\begin{proof}
  Assume $(N_1, m_1) \reduc_E (N_2, m_2)$ and property $F_1$ is an invariant on
  $N_1$.  Consider $m_2'$ a reachable marking in $N_2$.  By definition of
  $E$-abstraction, we have at least one reachable marking $m_1'$ in $N_1$ such
  that $m_1' \uplus m_2' \models E$.  Since $F_1$ is an invariant on $N_1$ we
  have $m_1' \models F_1$.  The condition $m_1' \uplus m_2' \models E$ is
  equivalent to $\underline{m_1'} \land \underline{m_2'} \land E$ satisfiable.
  By definition we have $\tilde{E}(\vec{p_1}, \vec{p_2}) \equiv E$, which
  implies $\underline{m_1'}(\vec{p_1}) \land \underline{m_2'}(\vec{p_2}) \land
  \tilde{E}(\vec{p_1}, \vec{p_2}) \land F_1(\vec{p_1})$ satisfiable, since the
  only variables that are both in $F_1$ and $E$ must also be in $N_1$. Hence,
  $m_2'$ satisfies the $E$-transform formula of $F_1$.
  The proof is similar in the other direction.
\end{proof}

Since $F_1$ invariant on $N_1$ is equivalent to $\neg F_1$ not
reachable, we can directly infer an equivalent conservation theorem
for reachability: to find a model of $F_1$ in $N_1$, it is enough to
find a model for
$F_1(\vec{p_1}) \wedge \tilde{E}(\vec{p_1}, \vec{p_2})$ in $N_2$.

\begin{theorem}[Reachability conservation]\label{th:formula_reachability_conservation}
  Assume $(N_1, m_1) \reduc_E (N_2, m_2)$ and that
  $F_2$ is the $E$-transform of formula $F_1$ on $N_1$. Then
  formula $F_1$ is reachable in $N_1$ if and only if
  $F_2$ is reachable in $N_2$.
\end{theorem}

\begin{proof}
  Assume $(N_1, m_1) \reduc_E (N_2, m_2)$ and property $F_1$ is reachable in
  $N_1$. Hence, there exists a reachable marking $m_1'$ in $N_1$ such that $m_1'
  \models F_1$. By definition of $E$-abstraction, we have at least one reachable
  marking $m_2'$ in $N_1$ such that $m_1' \uplus m_2' \models E$. The condition
  $m_1' \uplus m_2' \models E$ is equivalent to $\underline{m_1'}(\vec{p_1})
  \land \underline{m_2'}(\vec{p_2}) \land E$ satisfiable. By definition we have
  $\tilde{E}(\vec{p_1}, \vec{p_2}) \equiv E$, which implies $\underline{m_1'}
  \land \underline{m_2'} \land \tilde{E}(\vec{p_1}, \vec{p_2}) \land
  F_1(\vec{p_1})$ satisfiable, since the only variables that are both in $F_1$
  and $E$ must also be in $N_1$. Hence, $m_2'$ satisfies the $E$-transform
  formula of $F_1$.
  The proof is similar in the other direction.
\end{proof}


\section{BMC and PDR implementation}
\label{sec:implementation}

We developed a prototype model checker that takes advantage of net
reductions. The tool includes two main verification procedures that
have been developed for generalized Petri nets. (No specific
optimizations are applied when we know the net is safe, like for
instance using Boolean formulas instead of QF-LIA.) These procedures
correspond to instantiations of the BMC and PDR methods for checking
general reachability properties on Petri nets. We sketch these two
procedures below.

\subsection{Encoding Petri nets semantics using Linear Integer Arithmetic}

Our approach is based on a revisit of the semantics of Petri nets
using Linear Integer Arithmetic formulas.

We already defined (Sect.~\ref{sec:petri-nets-linear}) a helper
formula, or \emph{operator}, $\mathrm{ENBL}_t(\vec{x})$ such that
$\mathrm{ENBL}_t(\vec{x}) \wedge \underline{m}(\vec{x})$ is true when $t$ is
enabled at $m$.
We can define, in the same way, a linear predicate to describe the relation
between the markings before and after some transition $t$ fires. To this end, we
use a vector $\vec{x}'$ of ``primed variables'' $(x'_1, \dots, x'_n)$, where
$x'_i$ will stand for the marking of place $p_i$ after a transition is fired.

With this convention, formula $\mathrm{FIRE}_t(\vec{x}, \vec{x}')$ is such that
$\mathrm{FIRE}_t(m, m')$ entails $m \xrightarrow{t} m'$ when $t$ is
enabled at $m$, or $m = m'$.

\medskip
With all these notations, we can define a predicate
$\mathrm{T}(\vec{x}, \vec{x}')$ that ``encodes'' the effect of firing
at most one transition in the net $N$. By construction, formula
$\underline{m}(\vec{x}) \land \mathrm{T}(\vec{x}, \vec{x}')
\land \underline{m'}(\vec{x}')$ is true when $m \trans{} m'$, or when
$m = m'$.
\begin{align}
  \mathrm{ENBL}_t(\vec{x}) &\ \defeq\  \textstyle \bigwedge_{i \in 1..n} \left ( x_i \geqslant
  {\pre(t)}(p_i) \right )\\
  \Delta_t(\vec{x}, \vec{x}') &\ \defeq\ \textstyle  \bigwedge_{i \in 1..n} (x_i' =
                                x_i + \post(t,p_i) -
                                \pre(t,
                                p_i)) \label{eq:delta}\\
  \mathrm{EQ}(\vec{x}, \vec{x}') &\ \defeq\ \textstyle \bigwedge_{i \in 1..n} x_i =
                                   x_i'\\
  \mathrm{FIRE}_t(\vec{x}, \vec{x}') &\ \defeq\ \mathrm{EQ}(\vec{x},
                                   \vec{x}') \vee \left
                                   ({\mathrm{ENBL}}_t(\vec{x})
                                   \land \Delta_t(\vec{x},
                                   \vec{x}') \right ) \\
  \mathrm{T}(\vec{x}, \vec{x}') &\ \defeq\ \textstyle \mathrm{EQ}(\vec{x}, \vec{x}')
                              \lor
                              \bigvee_{t \in T} \left( \mathrm{ENBL}_t(\vec{x}) \land
                              \Delta_t(\vec{x}, \vec{x}') \right )
\end{align}

\subsection{Bounded Model Checking (BMC)}

BMC is an iterative method for exploring the state space of
finite-state systems by unrolling their
transitions~\cite{biere_symbolic_1999}. The method was originally
based on an encoding of transition systems into (a family of)
propositional logic formulas and the use of SAT solvers to check these
formulas for satisfiability~\cite{clarke_bounded_2001}.  There are
several works that adapt BMC to Petri nets, such
as~\cite{heljanko2001bounded}. More recently, this approach was
extended to more expressive models, and richer theories, using SMT
solvers~\cite{armando_bounded_2006}. Our goal is not to develop a
state of the art BMC model checker for generalized Petri
nets. Instead, we develop a textbook implementation that is enough to
test the impact on performances when using reductions.

\subsubsection{Description of the algorithm}

In BMC, we try to find a reachable marking $m$ that is a model for a given
formula $F$, that usually models a set of ``feared events''. The algorithm (see
function \hyperref[fun:bmc]{\texttt{BMC}}) starts by computing a formula, say $\phi_0$,
representing the initial marking and checking whether $\phi_0 \land F$ is
satisfiable (meaning $F$ is initially true). If the formula is unsatisfiable,
we compute a formula $\phi_1$ representing all the markings reachable in one
step, or less, from the initial marking and check $\phi_1 \land F$. This way, we
compute a sequence of formulas $(\phi_i)_{i \in \Nat}$ until either $\phi_i
\land F$ is satisfiable (in which case a counter-example is found) or we have
$\phi_{i+1} \Rightarrow \phi_i$ (in which case we reach a fixed point and no
counter-example exists).
The BMC method is not complete since it is not possible, in general,
to bound the number of iterations needed to give an answer. Also, when
the net is unbounded, we may very well have an infinite sequence of
formulas $\phi_0 \subsetneq \phi_1 \subsetneq \dots$ However, in
practice, this method can be very efficient to find a counter-example
when it exists.

The crux of the method is to compute formulas $\phi_i$ that represent
the set of markings reachable using firing sequences of length at most
$i$. We show how we can build such formulas incrementally. We assume
that we have a marked net $(N, m_0)$ with places
$P = \{p_1, \dots, p_n\}$ and transitions $T = \{t_1, \dots, t_k\}$.
In the remainder of this section, we build formulas that express
constraints between markings $m$ and $m'$ such that $m \rightarrow m'$
in $N$. Hence we define formulas with $2\, n$ variables. We use the
notation $\psi(\vec{x}, \vec{x}')$ as a shorthand for
$\psi(x_1, \dots, x_n, x_1', \dots, x_n')$.

\medskip
Formula $\phi_i$ is the result of connecting $i$ successive
occurrences of formulas of the form
$\mathrm{T}(\vec{x}_j, \vec{x}_{j+1})$. We define the formulas
inductively, with a base case ($\phi_0$) which states that only $m_0$
is reachable initially. To define the $\phi_i$'s, we assume that we
have a collection of (pairwise disjoint) sequences of variables,
$(\vec{x}_i)_{i \in \Nat}$. In the following listings, we use the
auxiliary function \texttt{freshVariables} to iterate over this
family of variable vectors.
\[
  \phi_0(N, m_0) \defeq \underline{m_0}(\vec{x}_0)\qquad
  \phi_{i+1}(N, m_0) \defeq \phi_i(N, m_0) \land \mathrm{T}(\vec{x}_{i},
                       \vec{x}_{i+1}) 
\]

\begin{function}[htbp]
  \DontPrintSemicolon

  \SetKwFunction{unsat}{unsat}
  \SetKwFunction{freshvariables}{freshVariables}

  \KwResult{$\top$ if $F$ is reachable (meaning $\neg F$ is not an invariant)}
  \BlankLine

  $\vec{x} \gets \freshvariables{}$\;
  $\phi \gets \underline{m_0}(\vec{x})$\;
  \BlankLine

  \While{\unsat{$\underline{m_0}(\vec{x}) \land \phi \land (\neg F)(\vec{x})$}}
  {
    $\vec{x}' \gets \freshvariables{}$\;
    $\phi \gets \phi \land T(\vec{x}, \vec{x}')$\;
    $\vec{x} \gets \vec{x}'$\;
    }
    \BlankLine

  \Return{$\top$}
  \caption{BMC($\mathrm{EF}\, F$: \protect{linear predicate})}
  \label{fun:bmc}
\end{function}

We can prove that this family of BMC formulas provide a way to check
 reachability properties, meaning that formula $F$ is reachable in $(N, m_0)$ if
 and only if there exists $i \geq 0$ such that $F(\vec{x}_i) \wedge \phi_{i}(N,
 m_0)$ is satisfiable. The approach we describe here is well-known (see for
 instance~\cite{biere_symbolic_1999}).
It is also quite simplified. Actual model checkers that rely on BMC apply
several optimization techniques, such as compositional reasoning; acceleration
methods; or the use of invariants on the underlying model to add extra
constraints. We do not consider such optimizations here, on purpose, since our
motivation is to study the impact of polyhedral abstractions. We believe that
our use of reductions is orthogonal and does not overlap with many of these
optimizations, in the sense that we do not preclude them, and that the
performance gain we observe with reductions could not be obtained with these \linebreak
optimizations.

\subsubsection{Combination with polyhedral abstraction}

Assume we have $(N_1, m_1) \reduc_E (N_2, m_2)$. We denote
$\mathrm{T}_1, \mathrm{T}_2$ the equivalent of formula $\mathrm{T}$,
above, for the nets $N_1, N_2$ respectively. We also use $\vec{x}$,
$\vec{y}$ for sequences of variables ranging over (the places of)
$N_1$ and $N_2$ respectively.  We should use $\phi(N_1, m_1)$ for the
family of formulas built using operator $\mathrm{T}_1$ and variables
$\vec{x}_0, \vec{x}_1, \dots$ and similarly for $\phi(N_2, m_2)$,
where we use $\mathrm{T}_2$ and variables of the form $\vec{y}_i$.

The following property states that, to find a model of $F$ in the reachable
markings of $N_1$ (meaning $\mathrm{EF}\, F$ true), it is enough to find a model
for its $E$-transform in $N_2$.

\begin{theorem}[BMC with $E$-transform]\label{th:bmcreduc}
  Assume $(N_1, m_1) \reduc_E (N_2, m_2)$ and that
  $F_2$ is the $E$-transform of $F_1$. Formula
  $F_1$ is reachable in $N_1$ if and only if there exists
  $j \geq 0$ such that $F_2(\vec{y}_j) \wedge \phi_j(N_2, m_2)$ is
  satisfiable.
\end{theorem}
\begin{proof}
  Our proof relies on the property that BMC is sound and complete for
  finding a finite counter-example (see
  e.g.\cite{cimatti_infinite-state_2016}): there is a firing sequence
  $\varrho$, of size less than $i$, such that
  $m_1 \wtrans{\varrho} m'_1$ and $m'_1 \models F_1$---meaning property
  $F_1$ is reachable in $N_1$--- if and only if
  $F_1 \wedge \phi_i(N_1, m_1)$. We can prove this property by
  induction on the value of $i$ and use the fact that $m \wtrans{} m'$
  or $m = m'$ in $N_1$ entails $\mathrm{T_1}(m, m')$.

  By our \emph{conservation of reachability} theorem
  (Th.~\ref{th:formula_reachability_conservation}), property $F_1$ is
  reachable in $N_1$ (say with a counter-example of size $i$) if and
  only if property $F_2$ is reachable in $N_2$ (say with a
  counter-example of size $j$). Therefore there exists $i$ such that
  $F_1(\vec{x}_i) \wedge \phi_i(N_1, m_1)$ is satisfiable if and only
  if there exists $j$ such that
  $F_2(\vec{y}_j) \wedge \phi_j(N_2, m_2)$ is satisfiable.
\end{proof}

We can give a stronger result, comparing the value of $i$ and $j$,
when the reductions used in proving the $E$-abstraction equivalence
never introduce new transitions.  This is the case, for example, with
the reductions computed using the \Reduce tool. Indeed, in this case,
we can show that we may find a witness of length $i$ in $N_1$ (a
firing sequence of length $i$ showing that $F_1$ is reachable in
$N_1$) when we find a witness of length $j \leq i$ in $N_2$. This is
because, in this case, reductions may compact a sequence of several
transitions into a single one or, at worst, not change it.  Take the
example of the [CONCAT] rule in Fig.~\ref{fig:reducs}. Therefore
BMC benefits from reductions in two ways. First because we can reduce
the size of formulas $\phi$ (which are proportional to the size of the
net), but also because we can accelerate transition unrolling in the
reduced net.

\subsection{Property Directed Reachability (PDR)}

While BMC is the right choice when we try to find counter-examples, it usually
performs poorly when we want to check an invariant property, $\mathrm{AG}\, F$.
There are techniques that are better suited to prove \emph{inductive invariants}
in a transition system; that is a property that is true initially and stays true
after firing any transition.

In order to check {invariants} with SMPT, we have implemented a method called
PDR~\cite{jhala_sat-based_2011,hutchison_understanding_2012} (also known as
IC3), which incrementally generates clauses that are inductive ``relative to
stepwise approximate reachability information''. PDR is a combination of
induction, over-approximation, and SAT solving. For SMPT, we developed a similar
method that uses SMT solving, to deal with markings and transitions, and that
can take advantage of polyhedral abstractions.

As with BMC, our focus is not to study the performances of PDR per se,
but its combination with polyhedral abstractions. We give more details
about our adaptation of PDR for generalized Petri nets
in~\cite{DBLP:conf/tacas/AmatDH22}. We keep enough information about
our implementation in order to state an equivalent of
Th.~\ref{th:bmcreduc} for PDR.

We use similar notations as with BMC, but with a small
difference. Indeed, since PDR ``unrolls at most one transition'' at a
time, we only need two vectors of variables instead of a family
$(\vec{x}_i)_{i \geqslant 0}$ like with BMC: we use unprimed variables
($\vec{x}$) to represent states before firing a transition and primed
variables ($\vec{x}'$) to represent the reached states.

PDR requires to define a set of \emph{safe states}, described as the
models of some property $F$. It also requires a set of initial states,
$I$. In our case $I \defeq \underline{m_0}$.

The procedure is complete for finite transition systems, for instance
with bounded Petri nets. We can also prove termination in the general
case when property $\neg F$ is \emph{monotonic}, meaning that
$m \models \neg F$ implies that $m'\models \neg F$ for all markings
$m'$ that cover $m$ (that is when $m' \geq m$, component-wise). An
intuition is that it is enough, in this case, to check the property on
the minimal coverability set of the net, which is always finite (see
e.g.~\cite{finkel1991minimal}).

A formula ${F}$ is \textit{inductive}~\cite{hutchison_understanding_2012} when
${I}(\vec{x}) \myplies {F}(\vec{x})$ and ${F}(\vec{x}) \land \mathrm{T}(\vec{x},
\vec{x}') \myplies {F}(\vec{x}')$ are tautologies. It is \textit{inductive
relative} to formula ${G}$ if both ${I}(\vec{x}) \myplies {F}(\vec{x})$ and
${G}(\vec{x}) \land {F}(\vec{x}) \land \mathrm{T}(\vec{x}, \vec{x}') \myplies
{F}(\vec{x}')$ are tautologies. With PDR we compute \textit{Over Approximated
Reachability Sequences} (OARS), meaning sequences of formulas $(F_0, \dots,
F_{k+1})$, with variables in $\vec{x}$, that are monotonic: $I(\vec{x}) \Rightarrow F_0(\vec{x})$,
$F_i(\vec{x}) \myplies F_{i+1}(\vec{x})$ for all $i \in 0..k$, and $F_{k+1} \myplies F$; and
satisfies \textit{consecution}: ${F_i}(\vec{x}) \wedge \mathrm{T}(\vec{x},
\vec{x}') \myplies {F_{i+1}}(\vec{x}')$ for all $i \leq k+1$. The formulas $F_i$
change at each iteration of the procedure (each time we increase $k$). The
procedure stops when we find an index $i$ such that $F_i = F_{i+1}$. In this
case we know that $F$ is an invariant. We can also stop during the iteration if
we find a counter-example.

\subsubsection{Description of the algorithm}

Our implementation follows closely the algorithm for IC3 described
in~\cite{hutchison_understanding_2012}. We only give the pseudo-code for the
four main functions (\hyperref[fun:prove]{\texttt{prove}}, \hyperref[fun:strengthen]{\texttt{strengthen}},
\hyperref[fun:inductively_generalize]{\texttt{inductivelyGeneralize}} and \hyperref[proc:push_generalization]{\texttt{pushGeneralization}}).\sloppy

The main function, \hyperref[fun:prove]{\texttt{prove}}, computes an \textit{Over Approximated
Reachability Sequence} (OARS) $(F_0, \allowbreak \dots, F_{k+1})$ of linear predicates,
called \emph{frames}, with variables in $\vec{x}$. An OARS meets the following
constraints: (1) it is monotonic: $F_i(\vec{x}) \land \neg F_{i+1}(\vec{x})$ is unsatisfiable, for
all $i \in 0..k$; (2) it contains the initial states: $I(\vec{x}) \land \neg F_0(\vec{x})$
is unsatisfiable; (3) it does not contain feared states: $F_{k+1}(\vec{x}) \land \neg F(\vec{x})$
is unsatisfiable; and (4) it satisfies consecution: ${F_i}(\vec{x}) \wedge
\mathrm{T}(\vec{x}, \vec{x}') \land \neg {F_{i+1}}(\vec{x}')$ is unsatisfiable for
all $i \in 0..k+1$.

By construction, each {frame} $F_i$ in the OARS is defined as a set of clauses,
$\CL({F_i})$, meaning that ${F_i}$ is built as a formula in CNF: $F_i =
\bigwedge_{\cl \in \CL(F_i)} \cl$. We also enforce that $\CL(F_{i+1}) \subseteq
\CL(F_i)$ for all $i \in 0..k$, which means that the monotonicity property
between frames is trivially ensured.

The body of function \hyperref[fun:prove]{\texttt{prove}} contains a \emph{main iteration}
(line~\ref{l:main_iteration}) that increases the value of $k$ (the number of
levels of the OARS). At each step, we enter a second, minor iteration
(line~\ref{alg:minor_iteration} in function \hyperref[fun:strengthen]{\texttt{strengthen}}), where we
generate new minimal inductive clauses that will be propagated to all the
frames. Hence both the length of the OARS, and the set of clauses in its frames,
increase during computation.
The procedure stops when we find an index $i$ such that $F_i = F_{i+1}$. In this
case we know that $F_i$ is an inductive invariant satisfying $F$. We can also
stop during the iteration if we find a counter-example (a model $m$ of $\neg F$).
In this case, we can also return a trace leading to $m$.

\begin{function}[tb]\small
  \DontPrintSemicolon

  \SetKwFunction{sat}{sat}
  \SetKwFunction{strengthen}{strengthen}
  \SetKwFunction{propagateclauses}{propagateClauses}

  \KwResult{$\top$ if $F$ is an invariant, otherwise $\bot$ (meaning $\neg F$ is reachable)} \BlankLine

  \If{\sat{$I(\vec{x}) \land T(\vec{x}, \vec{x}') \land (\neg F)(\vec{x}')$}}
  {\KwRet{$\bot$}} \BlankLine

  $k \gets 1$,  $F_0 \gets I$, $F_1 \gets F$\;
  \BlankLine

  \While{$\top$\label{l:main_iteration}}{\If{\strengthen{$k$}}{\KwRet{$\bot$}}
    \BlankLine
    \propagateclauses{$k$}\;
    \BlankLine

    \If{$\CL(F_i) = \CL(F_{i+1})$ \textbf{for some} $1 \leqslant  i \leqslant
      k$}{\KwRet{$\top$}}
    \BlankLine
    $k \gets k + 1$\;}

  \caption{prove($\mathrm{AG}\, F$: \protect{linear predicate})}
  \label{fun:prove}
\end{function}

When we start the first minor iteration, we have $k = 1$, $F_0 = I$ and $F_{1} =
F$. If we have $F_k(\vec{x}) \wedge T(\vec{x}, \vec{x}') \wedge (\neg
F)(\vec{x}')$ unsatisfiable, it means that $F$ is inductive, so we can stop and
return that $F$ is an invariant. Otherwise, we proceed with the strengthen
phase, where each model of $F_k(\vec{x}) \wedge T(\vec{x}, \vec{x}') \wedge
(\neg F)(\vec{x}')$ becomes a potential counter-example, or \emph{witness}, that
we need to ``block'' (line $3$--$5$ of function \hyperref[fun:strengthen]{\texttt{strengthen}}).

Instead of blocking only one witness (and to overcome the problem with a
potential infinite number of witnesses), we first generalize it into a predicate
that abstracts similar dangerous states (see the calls to
\texttt{generalizeWitness}). We define the formula
$\underline{\hat{m}}(\vec{x}) \defeq \textstyle \bigwedge_{i \in 1..n} (x_i \ge
m(p_i))$ that is valid for every marking that covers $m$; in the sense that
${m'} \models \underline{\hat{m}}$ only when $m'\geq m$.  By virtue of the
monotonicity of the flow function of Petri nets, when $\neg F$ is monotonic and
$m$ is a witness, we know that all models of $\underline{\hat{m}}$ are also
witnesses. Hence we can improve the method by generating a \emph{minimal inductive
clause} (MIC) from $\neg \underline{\hat{m}}$ instead of $\neg
\underline{m}$. Another benefit of this choice is that
$\underline{\hat{m}}(\vec{x})$ is a conjunction of inequalities of the form $(x_j \geq
k_i)$, which greatly simplifies the computation of the MIC. When $F$ is
anti-monotonic ($\neg F$ is monotonic), we can prove the completeness of the
procedure using an adaptation of Dickson's lemma, which states that we cannot
find an infinite decreasing chain of witnesses (but the number of possible
witness may be extremely large).
Hence, when we block it, we learn new clauses from $\neg \underline{\hat{m}}$ that can be
propagated to the previous frames.

\begin{function}[tb]\small
  \DontPrintSemicolon

  \SetKwProg{try}{try}{:}{} \SetKwProg{catch}{catch}{:}{end}

  \SetKwFunction{sat}{sat}
  \SetKwFunction{generalizedwitness}{generalizeWitness}
  \SetKwFunction{inductivelygeneralize}{inductivelyGeneralize}
  \SetKwFunction{pushgeneralization}{pushGeneralization}

  \try{}
  {\While{$(m \trans{t} m') \models\ F_k(\vec{x}) \land T(\vec{x}, \vec{x}') \land (\neg F)(\vec{x}')$\label{alg:minor_iteration}}
    {
      \textcolor{red}{$\hat{m} \gets \generalizedwitness{m}$}\;\label{l:strengten_witness}
      $n \gets$ \inductivelygeneralize{$\hat{m}, k - 2, k$}\;
      \pushgeneralization{$\{(\hat{m}, n+1)\}, k$}\;
    }
    \KwRet{$\top$}}

    \BlankLine

  \catch{\text{counter example}}{\KwRet{$\bot$}}

  \caption{strengthen($k$ : \protect{current level})}
  \label{fun:strengthen}
\end{function}

Before pushing a new clause, we test whether $\underline{\hat{m}}$ is reachable from
previous frames. We take advantage of this opportunity to find if we have a
counter-example and, if not, to learn new clauses in the process. This is the
role of functions \hyperref[proc:push_generalization]{\texttt{pushGeneralization}} and
\hyperref[fun:inductively_generalize]{\texttt{inductivelyGeneralize}}.

The goal of \hyperref[fun:inductively_generalize]{\texttt{inductivelyGeneralize}} is to strengthen the invariants in
$(F_i)_{i \leqslant k}$ by learning new clauses and finding the smallest index
in $1..k$ that can lead to the dangerous states $\hat{m}$. The goal of
\hyperref[proc:push_generalization]{\texttt{pushGeneralization}} is to apply inductive generalization, starting
from the earliest possible level.

We find a counter example (in the call to \hyperref[fun:inductively_generalize]{\texttt{inductivelyGeneralize}}) if
the generalization from a witness found at level $k$, say $\hat{s}$, reaches
level $0$ and $F_0(\vec{x}) \land T(\vec{x}, \vec{x}') \land \hat{s}(\vec{x}')$
is satisfiable (line $1$ in \hyperref[fun:inductively_generalize]{\texttt{inductivelyGeneralize}}). Indeed, it means
that we can build a trace from $I$ to $\neg F$ by going through $F_1, \dots,
F_k$.

The method relies heavily on checking the satisfiability of linear formulas in
QF-LIA, which is achieved with a call to a SMT solver. In each function call, we
need to test if predicates of the form $F_i(\vec{x}) \land T(\vec{x}, \vec{x}') \land G(\vec{x}')$ are unsatisfiable
and, if not, enumerate its models. To accelerate the strengthening of frames, we
also rely on the unsat core of properties in order to compute\linebreak MIC.

\medskip
\begin{function}[tb]\small
  \DontPrintSemicolon

  \SetKwFunction{sat}{sat} \SetKwFunction{generateclause}{generateClause}

  \If{$min < 0$ \textbf{and} \sat{$F_0(\vec{x}) \land T(\vec{x}, \vec{x}')
   \land s(\vec{x}')$}} {\textbf{raise} $Counterexample$}
  \BlankLine

  \For{$i \gets \max(1, min + 1)$ \textbf{to} $k$} {\If{\sat{$F_i(\vec{x}) \land
  T(\vec{x}, \vec{x}') \land \neg s(\vec{x}) \land s(\vec{x}')$}}
  {$\generateclause{s, i-1, k}$\;\KwRet{$i - 1$}}}
  \BlankLine

  $\generateclause{s, k, k}$\;
  \BlankLine

  \KwRet{$k$}

  \caption{inductivelyGeneralize($s$ : cube,  $min$: level, $k$: \protect{level})}
  \label{fun:inductively_generalize}
\end{function}

\begin{function}[tb]\small
  \DontPrintSemicolon

  \SetKwFunction{sat}{sat}
  \SetKwFunction{generalizedwitness}{generalizeWitness}
  \SetKwFunction{inductivelygeneralize}{inductivelyGeneralize}

  \While{$\top$} {
    $(s,n) \gets \text{from } states \text{ minimizing } n$\;
    \BlankLine

    \If{$n > k$} {\KwRet{}}
    \BlankLine

    \If{$(m \trans{t} m') \models F_n(\vec{x}) \land T(\vec{x}, \vec{x}') \land s(\vec{x}')$}
    {
    \textcolor{red}{$\hat{m} \gets \generalizedwitness{m}$}\;\label{l:pushgeneralization_witness}
    $l \gets $\inductivelygeneralize{$\hat{m}, n - 2, k$}\;
    $states \gets states \cup \{(p, l + 1)\}$\;
    }
    \Else{
      $l \gets \inductivelygeneralize{s, n, k}$\;
      $states \gets states \setminus \{(s, n)\} \cup \{(s, l + 1)\}$\;
    }
  }

  \caption{pushGeneralization($states$: \text{set of} (state, level), $k$: \protect{level})}
  \label{proc:push_generalization}
\end{function}

\subsubsection{Combination with polyhedral abstraction}

Assume we have $(N_1, m_1) \reduc_E (N_2, m_2)$ and that $G_2$ is the
$E$-transform of formula $G_1$ on $N_1$. We also assume that $G_1$ and
$G_2$ are anti-monotonic (meaning $\neg G_1$ and $\neg G_2$ monotonic), in order
to ensure the termination of the PDR procedure.  (We can prove that $\tilde{E}$
is monotonic for systems $E$ computed with the \Reduce tool when the initial net
does not use inhibitor arcs.) To check that formula $G_1$ is an invariant on
$N_1$ (meaning $\mathrm{AG}\, G_1$ true), it is
enough~\cite{jhala_sat-based_2011} to incrementally build OARS $(F_0, \dots,
F_{k+1})$ on $N_1$ until $F_i = F_{i+1}$ for some index $i \in 0..k$. In this
context, $F_0 = \underline{m_1}$ and $F_{k+1} \Rightarrow G_1$. In a similar way
than with our extension of BMC with reductions, a corollary of our
\emph{invariant conservation} theorem (Th.~\ref{th:invariant_conservation}) is
that, to check that $G_1$ is an invariant on $N_1$, it is enough to build OARS
$(F'_0, \dots, F'_{l+1})$ on $N_2$ where $F'_0 = \underline{m_2}$ and $F'_{l+1}
\Rightarrow G_2$.

\begin{theorem}[PDR with $E$-transform]\label{th:pdrreduc} Assume $(N_1,
  m_1) \reduc_E (N_2, m_2)$ and that $G_2$ is the $E$-transform of
  $G_1$, both anti-monotonic formulas. Formula $G_1$ is an invariant on
  $N_1$ if and only if there exists $i\geq 0$ such that $F'_i = F'_{i+1}$ in the
  OARS built from net $N_2$ and formula $G_2$.
\end{theorem}

\begin{proof}
  Our proof relies on the property that PDR is sound and complete for
  finite-state systems, see for instance Th.~1 in
  \cite{jhala_sat-based_2011}: $G_1$ is an invariant on $(N_1, m_1)$
  if and only if there exists $i\geq 0$ such that $F_i =
  F_{i+1}$. Since we deal with monotonic formulas (in this case the feared states $\neg G_1$ and $\neg G_2$), the set of markings
  satisfying a frame $F_i$ is always upward-closed (if $m \models F_i$
  and $m'\geqslant m$ then also $m'\models F_i$), and therefore we can
  work on the \emph{coverability graph} of the net, instead on its
  actual marking graph, which is finite~\cite{finkel1991minimal}. The
  results follows from our \textit{invariant conservation} theorem
  (Th.~\ref{th:invariant_conservation}).
\end{proof}

\subsection{Combination of BMC and PDR}

In the next section (Sect.\ref{sec:experimental-results}), we report on the
results obtained with our implementation of BMC and PDR (with and without
reductions), on an independent and comprehensive set of benchmarks.

With PDR, we restrict ourselves to the proof of liveness properties,
$\mathrm{EF}\, \phi$ where $\phi$ is monotonic (or equivalently, invariants
$\mathrm{AG}\, \phi$ with $\phi$ anti-monotonic). In practice, we do not check
if $\phi$ is monotonic using our ``semantical'' definition. Instead, our
implementation uses a syntactical restriction that is a sufficient condition for
monotonicity. This is the case, for example, when testing the quasi-liveness of
a set of transitions. On the other hand, deadlock is not monotonic. In such
cases, we can only rely on the BMC procedure, which may not terminate if the net
has no deadlocks.  Hence, our best-case scenario is when we check a monotonic
property (or if a model for the property exists). In our benchmarks, we find
that almost $30\%$ of all the properties are monotonic.

We have plans to improve our PDR procedure to increase the set of properties
that can be handled. In particular, we know how to do better when the net is
$k$-bounded (and we know the value of $k$). We also have several proposals to
improve the computation of a good witness, and its MIC, in the general case. We
should explore all these ideas in a future work.


\section{Experimental results}
\label{sec:experimental-results}

We have implemented the approach described in
Sect.~\ref{sec:implementation} into a new tool, called {SMPT} (for
{Satisfiability Modulo P/T Nets}). The tool is open-source, under the
GPLv3 license, and is freely available on {GitHub}
(\url{https://github.com/nicolasAmat/SMPT/}). In this section, we
report on some experimental results obtained with SMPT (v2) on an
extensive benchmark of models and formulas provided by the Model
Checking Contest (MCC)~\cite{mcc2019,HillahK17}.

{SMPT} serves as a front-end to generic SMT solvers, such as
{z3}~\cite{de_moura_z3_2008,z3}. The tool can output sets of
constraints using the SMT-LIB format~\cite{BarFT-RR-17} and pipe them
to a {z3} process through the standard input. We have implemented our
tool with the goal to be as interoperable as possible, but we have not
conducted experiments with other solvers yet.

SMPT takes as inputs Petri nets defined using the \texttt{.net} format
of the TINA toolbox and can therefore also accept nets defined using
the PNML syntax. For formulas, we accept properties defined with the
XML syntax used in the MCC competition. The tool does not compute net
reductions directly but relies on the tool \Reduce, that we described
at the end of Sect.~\ref{sec:reach-equiv-net}.

The tool \Reduce is distributed with the standard distribution of the TINA
toolbox, which includes the most mature versions of our verification tools. We
also provide an open source, feature complete version of an equivalent tool,
called \textsc{Shrink} (\url{https://github.com/Fomys/pnets}), that provide
several Rust libraries for manipulating Petri nets and performing structural
reductions.

\subsection{Benchmarks and distribution of reduction ratios}

Our benchmark suite is built from a collection of $102$ models used in the MCC
competition. Most of the models are parametrized, and therefore there can be
several different \emph{instances} for the same model. There are about $1\,000$
different instances of Petri nets whose size vary widely, from $9$ to $50\,000$
places, and from $7$ to $200\,000$ transitions. Most nets are ordinary (non-zero
weights on all arcs are equal to $1$), but a significant number of instances (about
$150$) use weighted arcs. Overall, the collection provides a large number of
examples with various structural and behavioral characteristics, covering a
large variety of use cases.

Since our approach relies on the use of net reductions, it is natural
to wonder if reductions occur in practice. To answer this question, we
computed the {reduction ratio} ($r$), obtained using \Reduce, as a
quotient between the number of places before ($p_\mathrm{init}$) and
after ($p_\mathrm{red}$) reduction:
$r = (p_\mathrm{init} - p_\mathrm{red}) / p_\mathrm{init}$. We display
the results for the whole collection of instances in
Fig.~\ref{fig:reduction}, sorted in descending order.

A ratio of $100\%$ ($r = 1$) means that the net is \emph{fully
  reduced}; the resulting net has only one (empty) marking. We see
that there is a surprisingly high number of models that are totally
reducible with our approach (about $20$\% of the total number), with
approximately half of the instances that can be reduced by a ratio of
$30\%$ or more.

\begin{figure}[htb]
  \centering
  \includegraphics[width=\linewidth]{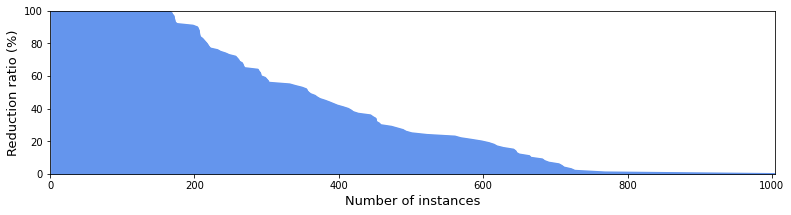}
  \caption{Distribution of reduction ratios over the instances in the MCC\label{fig:reduction}}
\end{figure}

For each edition of the MCC, a collection of about $30$ random
reachability properties are generated for each instance. We evaluated
the performance of SMPT using the formulas of the MCC'2020, on a
selection of $426$ Petri nets taken from instances with a reduction
ratio greater than $1\%$. (To avoid any bias introduced by models with
a large number of instances, we selected at most 5 instances with a
similar reduction ratio from each model.)

A pair of an instance and a formula is called a \emph{test case}. For
each test case, we check the formulas with and without the help of
reductions (using both the BMC and PDR methods in parallel) and with a
fixed timeout of \SI{120}{\second}. This adds up to a total of
$13\,265$ \emph{test cases} which required the equivalent of $447$
hours of CPU time.

\subsection{Impact on the number of solvable queries}

We report our results in the table of Fig.~\ref{fig:table}.  Out of
the almost $13\,000$ test cases, we were able to compute approximately
$7\,000$ results using reductions and only $3\,555$ without reductions
(so approximately twice more).

We compared our results with the ones provided by an
\emph{oracle}~\cite{oracle}, which gives the expected answer (as
computed by a majority of tools, using different techniques, during
the MCC competition). We achieve $100\%$ reliability on the benchmark;
meaning we always give the answer predicted by the oracle.


We give the number of computed results for four different categories
of test cases: \emph{Full} contains only the fully reducible instances
(the best possible case with our approach); while
\emph{Low}/\emph{Good}/\emph{High} correspond to instances with a
low/moderate/high level of reduction. We chose the limits for these
categories in order to obtain samples with comparable sizes.
We also have a general category,
\emph{All}, for the complete set of benchmarks.

\newcolumntype{x}[1]{>{\centering\arraybackslash\hspace{0pt}}p{#1}}
\begin{figure}[htb]
\vspace*{1mm}
  \begin{center}
    \begin {tabular}{l@{\ }l  x{0em} r x{0em}  rr x{0.5em} rr}%
      \toprule
      \multicolumn{2}{c}{\multirow{2}{*}{\begin{minipage}[c]{5em}
            \centering\textsc{Reduction Ratio ($r$)}
          \end{minipage}}} & \multicolumn{3}{c}{\multirow{2}{*}{\begin{minipage}[c]{4em}
                                \centering \textsc{\# Test\\ Cases}
                              \end{minipage}}} &
                                                 \multicolumn{5}{c}{\textsc{Results
                                                 (BMC/PDR)}}\\\cmidrule(rl){6-10}
	
      \multicolumn{2}{c}{}
                           &&&
                                              & \multicolumn{2}{c}{\textsc{With reductions}}
                           && \multicolumn{2}{c}{\textsc{Without}}\\\midrule
	
      \emph{All} & $r \in \; ]0, 1]$   &&  $13\,265$ && $6\,986$&
                           && $3\,555$ & ($3\,261$\,/\,$294$)\\
      \emph{Low} & $r \in \; ]0, 0.25[$   && $4\,586$ && $1\,662$& ($1\,532$\,/\,$130$)
                           && $1\,350$ & ($1\,247$\,/\,$103$)\\
      \emph{Good} & $r \in [0.25, 0.5[$   && $2\,823$ && $1\,176$
                                              &($1\,084$\,/\,$92$)  && $704$ & ($631$\,/\,$73$)\\
      \emph{High} & $r \in [0.5, 1[$   && $3\,298$ && $1\,591$ & ($1\,412$\,/\,$179$)   &&
                                                                                        $511$ & ($457$\,/\,$54$)\\
      \emph{Full} & $r = 1$ && $2\,558$ && $2\,557$ &  && $990$ & ($926$\,/\,$64$) \\\bottomrule
    \end{tabular}
  \end{center}\vspace*{-2mm}
\caption{Impact of the reduction ratio on the number of solved instances.\label{fig:table}}
\end{figure}

We observe that we are able to compute almost twice as many results
when we use reductions than without. This gain is greater on the
\emph{High} ($\times 3.1$) than on the \emph{Good} ($\times 1.7$)
instances. Nonetheless, the fact that the number of additional queries
solved using reductions is still substantial, even for a reduction
ratio under $50\%$, indicates that our approach can benefit from all
the reductions we can find in a model (and that our results are not
skewed by the large number of fully reducible instances).

In the special case of \textit{fully reducible} nets, checking a query
amounts to solving a linear system on the initial marking of the
reduced net. There are no iterations. Moreover this is the same system
for both the BMC and PDR procedures. For this category, we are able to
compute a result for all but one of the queries (that could be
computed using a timeout of $\SI{180}{\second}$). Most of these
queries can be solved in less than a few seconds.

When the distinction makes sense, we also report the number of cases
solved using BMC/PDR. (As said previously, the two procedures coincide
in category \emph{Full}, with reductions.) We observe that the
contribution of PDR is poor. This can be explained by several
factors. First, we restricted our implementation of PDR to monotonic
formulas (which represents $30\%$ of all properties). Among these, PDR
is useful only when we have an invariant that is true (meaning BMC
will certainly not terminate). On the other hand, PDR is able to give
answers on the most complex cases. Indeed, it is much more difficult
to prove an invariant than to find a counter-example (and we have
other means to try and find counter-examples, like simulation for
instance). This is why we intend to improve the performances and the
``expressiveness'' of our PDR implementation. Another factor, already
observed in~\cite{thierry-mieg_structural_2020}, is the existence of a
bias in the MCC benchmark: in more than $60\%$ of the cases, the
result follows from finding a counter-example (meaning an invariant
that is false or a reachability property that is true).

\subsection{Impact on computation time}

To better understand the impact of reductions on the computation time,
we compare the computation time, with or without reductions, for
each test case. These results do not take into account the time spent
for {reducing} each instance. This time is negligible when compared to
each test, usually in the order of $\SI{1}{\second}$. Also, we only need to reduce the net
once when checking the $30$
properties for the same instance.

\medskip
We display our results in Fig.~\ref{fig:solving_time}, where we give
four scatter plots comparing the computation time ``with'' ($y$-axis)
and ``without'' reductions ($x$-axis), for the \emph{Low},
\emph{Good}, \emph{High} and \emph{Full} categories of instances. Each
chart uses a logarithmic scale. We also display a histogram, for each
axis on the charts, that gives the density of points for a given
duration.  To avoid overplotting, we removed all the ``trivial''
properties (the bottom left part of the chart), that can be computed
with and without reduction in less than $\SI{10}{\milli\second}$.
These ``trivial'' queries ($507$ in total) correspond to instances with a small state
space or to situations where a counter-example can be found very
quickly.

\begin{figure}[htbp]
	\centering

	{
    \begin{subfigure}[b]{.49\textwidth}
			\centering
			\includegraphics[width=\textwidth]{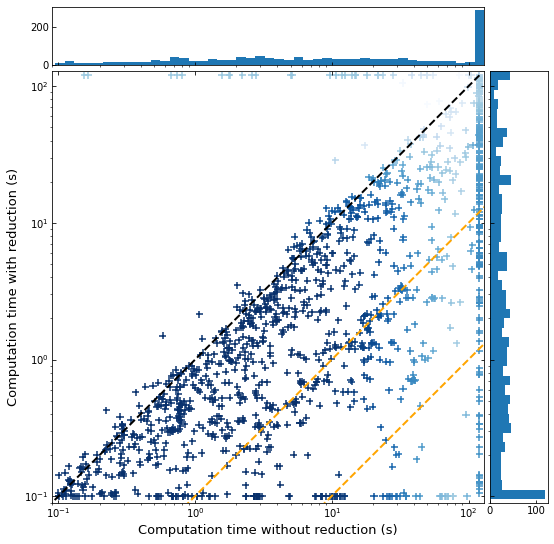}
      \caption{$r \in ]0, 0.25[$}
    \end{subfigure}
		\begin{subfigure}[b]{.49\textwidth}
			\centering
			\includegraphics[width=\textwidth]{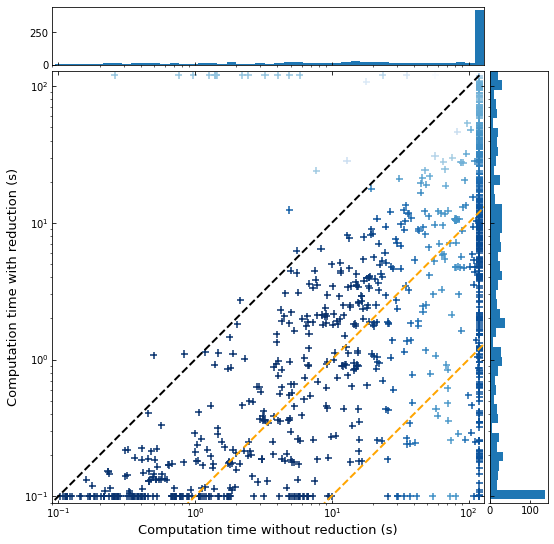}
      \caption{$r \in [0.25, 0.5[$}
		\end{subfigure}
		\begin{subfigure}[b]{.49\textwidth}
			\centering
			\includegraphics[width=\textwidth]{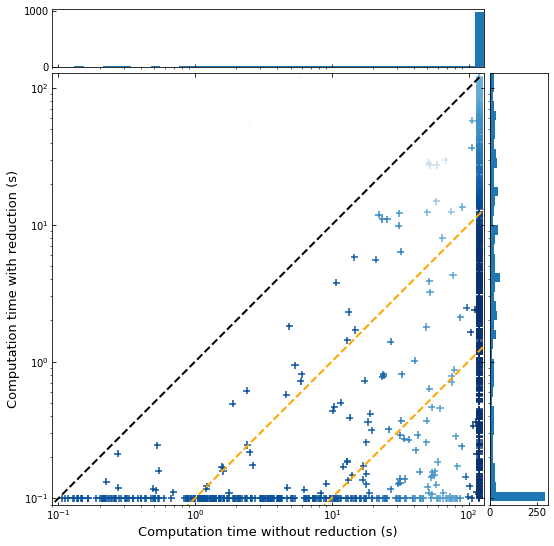}
      \caption{$r \in [0.5, 1[$}
    \end{subfigure}
		\begin{subfigure}[b]{.49\textwidth}
			\centering
			\includegraphics[width=\textwidth]{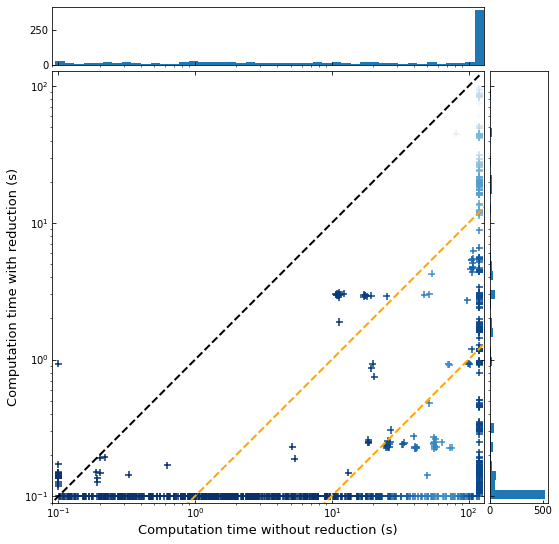}
      \caption{$r = 1$}
    \end{subfigure}}
	
	\caption{Comparing computation time, ``with'' ($y$-axis) and
          ``without'' ($x$-axis) reductions for categories \emph{Low} (a),
          \emph{Good} (b), \emph{High} (c)  and \emph{Full} (d).}
	
	\label{fig:solving_time}
\end{figure}

We observe that almost all the data points are below the diagonal,
meaning reductions accelerate the computation, with many test cases
exhibiting speed-ups larger than $\times 100$. We have added two
light-coloured, dashed lines to materialize data points with speed-ups
larger than $\times 10$ and $\times 100$ respectively.

On our $13\,265$ test cases, we timeout with reductions but compute a
result without on only $51$ cases ($0.4\%$). These exceptions can be
explained by border cases where the order in which transitions are
processed has a sizeable impact.

Another interesting point is the ratio of properties that can be
computed only using reductions. This is best viewed when looking at
the histogram values. A vast majority of the points in the charts are either on
the right border (computation without reductions timeout) or on the
$x$-axis (they can be computed in less than $\SI{10}{\milli\second}$
using reductions).


\section{Related work and conclusion}

We propose a new method to combine structural reductions with SMT
solving in order to check invariants on arbitrary Petri nets. While
this idea is not original, the framework we developed is new. Our main
innovation resides in the use of a principled approach, where we can
trace back reachable markings (between an initial net and its
residual) by means of a conjunction of linear equalities (the formula
$\tilde{E}$).
Basically, we show that we can adapt a SMT-based procedure for
checking a property on a net (that relies on computing a family of
formulas of the form $(\phi_i)_{i \in I}$) into a procedure that
relies on a reduced version of the net and formulas of the form
$(\phi_i \wedge \tilde{E})_{i \in J}$.

As a proof of concept, we apply our approach to two basic
implementations of the BMC and PDR procedures. Our empirical
evaluation shows promising results. For example, we observe that we
are able to compute twice as many results using reductions than
without. We believe that our approach can be adapted to more decision
procedures and could easily accommodate various types of
optimizations.

\subsection{Related work}

Our main theoretical results (the conservation theorems of
Sect.~\ref{sec:smt-based-model}) can be interpreted as examples of
\emph{reduction
  theorems}~\cite{lipton_reduction_1975,lamport_reduction_2016}, that
allow to deduce properties of an initial model ($N$) from properties
of a simpler, coarser-grained version ($N^R$). While these works are
related, they mainly focus on reductions where one can group a
sequence of transitions into a single, atomic action.
Hence, in our context, they correspond to a restricted class of
reductions, similar to a subset of the agglomeration rules used
in~\cite{berthomieu_counting_2019}.

We can also mention approaches where the system is simplified with
respect to a given property, for instance by eliminating parts that
cannot contribute to its truth value, like with the slicing or
\emph{Cone of Influence} abstractions~\cite{clarke_model_1999} used in
some model checkers. Finding such ``parts'' (places and transitions)
in a Petri net is not always easy, especially when the formula
involves many places.  This is not a problem with our approach, since
we can always abstract away a place, as long as its effect is
preserved in the $E$-transform formula.

In~\cite{schmidt_tacas_2003}, the author uses net invariants to
compress the representation of markings. This approach is based on the
fact that place invariants provide linear constraints between the
markings of several places; like in our use of redundancy rules. But
the goal is to reduce the memory footprint when computing the explicit
state space, while verification is still performed on ``uncompressed''
markings. On the contrary, our approach can be used with symbolic
methods--working on a reduced version of the net--and can use more
general rules. For instance, it cannot benefit from rules that
agglomerate places.

In practice, we derive polyhedral abstractions using \emph{structural
  reductions}, a concept introduced by Berthelot
in~\cite{berthelot_transformations_1987}. In our work, we are
interested in reductions that preserve the reachable states. This is
in contrast with most works about reductions, where more powerful
transformations can be applied when we focus on specific properties,
such as the absence of deadlocks. Several tools use reductions for
checking reachability properties. TAPAAL~\cite{bonneland2019stubborn},
for instance, is an explicit-state model checker that combines
Partial-Order Reduction techniques and structural reductions and can
check properties on Petri nets with weighted arcs and inhibitor arcs.

A more relevant example is ITS
Tools~\cite{thierry-mieg_structural_2020,DBLP:journals/fuin/Thierry-Mieg21}, which combines several
techniques, including structural reductions and the use of SAT and SMT
solvers. This tool relies on efficient methods for finding
counter-examples---with the goal to invalidate an invariant---based on
the collaboration between pseudo-random exploration techniques; hints
computed by an SMT engine; and reductions that may simplify atoms in
the property or places and transitions in the net. It also describes a
semi-decision procedure, based on an over-approximation of the
state space, that may detect when an invariant holds (by ruling out
infeasible behaviours). This leads to a very efficient tool, able to
compute a result for most of the queries in our benchmark, when we
solve only $52\%$ of our test cases. Nonetheless, we are able to solve
$46$ queries with SMPT (with a timeout of \SI{120}{\second}) that are
not in the oracle results collected from ITS Tools~\cite{oracle};
which means that no other tool was able to compute a result on these queries.

It has to be kept in mind, though, that our goal is to study the
impact of polyhedral abstraction, in isolation from other
techniques. However, the methods described
in~\cite{thierry-mieg_structural_2020} provide many ideas for
improving our approach, such as: using linear arithmetic over
reals---which is more tractable than integer arithmetic---to
over-approximate the state space of a net; adding extra constraints to
strengthen invariants (for instance using the state equation or
constraints derived from traps); dividing up a formula into smaller
sub-parts, and checking them incrementally or separately; \dots\ But
the main lesson to be learned is that there is a need for a complete
decision procedure devoted to the proof of satisfiable invariants,
which further our interest in improving our implementation of PDR.

A byproduct of our work is to provide a partial implementation
of PDR that is correct and complete when the property is monotonic
(see Sect.~\ref{sec:smt-based-model}), even in the case of nets with
an infinite state space. Our current solution
can be understood as a restriction to the case of ``coverability properties'',
which seems to be the current state-of-the-art with Petri nets; see for
example~\cite{esparza_smt-based_2014} and
\cite{DBLP:conf/apn/KangBJ21}, or the extension of PDR to
``well-structured transition systems''~\cite{kloos_incremental_2013}.
The reachability problem for Petri nets or, equivalently, for Vector Addition
Systems with States (VASS) is decidable~\cite{kosaraju}. Even if this result is
based on a constructive proof, and its ``construction'' streamlined over
time~\cite{leroux2009general}, the classical
Kosaraju-Lambert-Mayr-Sacerdote-Tenney approach does not lead to a workable
algorithm. It is in fact a feat that this algorithm has been implemented at all,
see e.g. the tool \textsc{KReach}~\cite{dixon_kreach_2020}.
While the (very high) complexity of the problem means that no single algorithm
could work efficiently on all inputs, it does not prevent the existence of
methods that work well on some classes of problems. For example, several
algorithms are tailored for the discovery of counter-examples. We can mention the
tool \textsc{FastForward}~\cite{blondin_directed_2021}, that
explicitly targets the case of unbounded nets.
We can also mention the works on inductive procedures for infinite-state and/or
parametrized systems, such as the verification methods used in
Cubicle~\cite{conchon_cubicle_2012}; see
also~\cite{cimatti_infinite-state_2016,gurfinkel2016smt}.

\subsection{Follow up work}

We propose a new method that adapts our approach---initially developed
for model checking with decision
diagrams~\cite{berthomieu2018petri,berthomieu_counting_2019}---for use
with SMT solvers.

We have continued working with our polyhedral abstraction since our
initial publication in~\cite{pn2021}. In particular, we tried applying
our approach to the verification of properties more complex than
reachability, like with our recent work on the \emph{concurrent
  places} problem~\cite{spin2021}. The problem, in this case, is to
enumerate all pairs of places that can be marked together, for some
reachable states. In this work we defined a new data-structure that
precisely captures the structure of reduction constraints, what we call
the \textit{Token Flow Graph} (TFG), and we used the TFGs to
accelerate the computation of the concurrency relation

We also continued improving our adaptation of PDR, which is the most
promising part of our work and raises several interesting theoretical
problems. In this context, we recently proposed two new adaptions of
PDR, to deal with non-monotonic
formulas~\cite{DBLP:conf/tacas/AmatDH22}. But there is still a lot of
work to be done, like for instance concerning the completeness of our
new approach and/or its limits.

\subsection{Future work}

There is still ample room for improving our tool. We already mentioned
some ideas for enhancements that we could borrow from ITS Tools, but
we also plan to specialize our verification procedures in some
specific cases, for example when we know that a net is $1$-safe. A
first step should be to compare our performances with other tools in
more details. This is what motivates our participation to the next
edition of the MCC, with SMPT alone in the reachability examinations,
even though it is common knowledge that winning tools need to combine
several different techniques.

On a more theoretical side, we also identified a need to develop an
automated (or a semi-automatic) method to prove the correctness of new
reduction rules.


\bibliographystyle{fundam}
\bibliography{bibfile}

\end{document}

%% file: macros.tex
\usepackage{amssymb}
\usepackage{mathtools}

\usepackage{graphicx}

\usepackage{booktabs}
\usepackage{array}
\usepackage{multirow,makecell}
\usepackage{pgfplots}
\usepackage{url}
\usepackage{float}
\usepackage{stmaryrd}

\usepackage{bookmark}

\usepackage{siunitx}
\usepackage{arydshln}
\usepackage{xspace}
\usepackage[linesnumbered,ruled,vlined,noend]{algorithm2e}

\newcommand{\fv}[1]{\ensuremath{\mathit{fv}({#1})}}

\newcommand{\CL}{\ensuremath{\mathrm{CL}}}
\newcommand{\cl}{\ensuremath{\mathit{cl}}}

\renewcommand{\geq}{\geqslant}
\renewcommand{\leq}{\leqslant}

\newcommand{\pre}{\text{Pre}\xspace}
\newcommand{\post}{\text{Post}\xspace}

\newcommand{\Nat}{\mathbb{N}}

\newcommand{\myplies}{\mathop{\Rightarrow}}

\newcommand{\Reduce}{\textsc{Reduce}\xspace}

\newcommand{\reduc}{\vartriangleright}
\newcommand*{\defeq}{\triangleq}

\newif\ifmarkingSep
\newcommand*{\marking}[1]{%
  \relax
  \markingSepfalse
  \markingScan#1\relax\relax
}
\newcommand{\markingScan}[2]{%
  \ifx\relax#1\empty
  \else
    \ifmarkingSep
      {\ }\relax
    \else
      \markingSeptrue
    \fi
    {#1}{\ast}{#2}\relax
    \expandafter\markingScan
  \fi
}

\makeatletter
\def \rightarrowfill{\m@th\mathord{\smash-}\mkern-6mu%
  \cleaders\hbox{$\mkern-2mu\mathord{\smash-}\mkern-2mu$}\hfill
  \mkern-6mu\mathord\rightarrow}
\makeatother

\makeatletter
\def \Rightarrowfill{\m@th\mathord{\smash-}\mkern-6mu%
  \cleaders\hbox{$\mkern-2mu\mathord{\smash-}\mkern-2mu$}\hfill
  \mkern-6mu\mathord\Rightarrow}
\makeatother

\makeatletter
\def \rightarrowfill{\m@th\mathord{\smash-}\mkern-6mu%
  \cleaders\hbox{$\mkern-2mu\mathord{\smash-}\mkern-2mu$}\hfill
  \mkern-6mu\mathord\rightarrow}
\makeatother

\makeatletter
\def \Rightarrowfill{\m@th\mathord{\smash=}\mkern-6mu%
  \cleaders\hbox{$\mkern-2mu\mathord{\smash=}\mkern-2mu$}\hfill
  \mkern-6mu\mathord\Rightarrow}
\makeatother

\makeatletter
\def \midrightarrowfill{\m@th\mathord{\smash{\raisebox{.2ex}{$\scriptscriptstyle\mid$}}\!\!\,-}\mkern-6mu%
  \cleaders\hbox{$\mkern-2mu\mathord{\smash-}\mkern-2mu$}\hfill
  \mkern-6mu\mathord\rightarrow}
\makeatother

\makeatletter
\def \midRightarrowfill{\m@th\mathord{\smash{\raisebox{.1ex}{$\scriptstyle\mid$}}\!\!\!=}\mkern-6mu%
  \cleaders\hbox{$\mkern-2mu\mathord{\smash=}\mkern-2mu$}\hfill
  \mkern-6mu\mathord\Rightarrow}
\makeatother

\newcommand{\overstackrel}[2]{\mathrel{\mathop{#1}\limits^{#2}}}
\newcommand{\trans}[1]{\mathbin{\smash[t]{\overstackrel{\rightarrowfill}{\ #1\ }}}}
\newcommand{\wtrans}[1]{\mathbin{\smash[t]{\overstackrel{\Rightarrowfill}{\ #1\ }}}}

\usepackage{tikz}
\usetikzlibrary{arrows}
\usetikzlibrary{decorations.markings}
\usetikzlibrary{shadows}
\usepackage{subcaption}
\tikzset{
  big stealth/.style={
    decoration={markings,mark=at position -(0.1pt) with {\arrow[scale=2*\scale]{stealth}}},
    postaction={decorate},
    shorten >=0.4pt}}
\tikzset{
  big ring/.style={
    decoration={markings,mark=at position -(0.1pt) with {\arrow[scale=1.5*\scale]{o}}},
    postaction={decorate},
    shorten >=8pt*\scale}}
\tikzset{
  big disc/.style={
    decoration={markings,mark=at position -(0.1pt) with {\arrow[scale=1.5*\scale]{*}}},
    postaction={decorate},
    shorten >=8pt*\scale}}
\tikzset{
  big box/.style={
    decoration={markings,mark=at position -(0.1pt) with {\arrow[scale=1.5*\scale]{open square}}},
    postaction={decorate},
    shorten >=8pt*\scale}}
\tikzset{
  big tile/.style={
    decoration={markings,mark=at position -(0.1pt) with {\arrow[scale=1.5*\scale]{square}}},
    postaction={decorate},
    shorten >=8pt*\scale}}

\tikzstyle{state}=[circle, very thick, fill, top color=white, bottom color=white, draw=black, minimum size=40pt, drop shadow]
\tikzstyle{place}=[circle, very thick, fill, top color=white, bottom color=white, draw=black, minimum size=40pt, drop shadow]
\tikzstyle{bplace}=[circle, very thick, fill, top color=cyan, fill
opacity=0.2, text opacity=1, bottom color=white, draw=cyan, minimum size=40pt]
\tikzstyle{vplace}=[circle, very thick, fill, top color=violet, fill
opacity=0.2, text opacity=1, bottom color=white, draw=violet, minimum size=40pt]
\tikzstyle{trans}=[rectangle, very thick, fill, top color=white, bottom color=white, draw=black, minimum size=32pt, drop shadow]
\tikzstyle{btrans}=[rectangle, very thick, fill, top color=cyan, fill
opacity=0.2, text opacity=1, bottom color=white, draw=cyan, minimum size=32pt]
\tikzstyle{vtrans}=[rectangle, very thick, fill, top color=violet, fill
opacity=0.2, text opacity=1, bottom color=white, draw=violet, minimum size=32pt]
\tikzstyle{bvtrans}=[rectangle, fill, bottom color=violet, top color=cyan, fill
opacity=0.2, text opacity=1, minimum size=32pt]
\tikzstyle{arc}=[thick, big stealth, black]
\tikzstyle{barc}=[thick, big stealth, cyan]
\tikzstyle{varc}=[thick, big stealth, violet]
\tikzstyle{read}=[thick, big disc, black]
\tikzstyle{inhibitor}=[thick, big ring, black]
\tikzstyle{stopwatch}=[thick, big tile, black]
\tikzstyle{stopwatchinhibitor}=[thick, big box, black]
\tikzstyle{priority}=[thick, big stealth, orange]
\tikzstyle{enabling}=[thick, big disc, orange]
\tikzstyle{disabling}=[thick, big ring, orange]
\tikzstyle{token}=[circle, fill, draw=black, minimum size=4pt]
\tikzstyle{glob-options}=[label
distance=6pt*\scalenodes*\scale,x=1pt,y=-1pt,scale=\scale,every
node/.style={transform shape}, baseline=(current bounding box.center)]
\tikzstyle{virtual}=[circle, draw=white, minimum size=20pt]

\def\scale{0.57}
\def\scalenodes{1.0}